\documentclass[a4paper,11pt]{article} 

\usepackage{fullpage}
\usepackage{times}
\usepackage{soul}
\usepackage{url}
\usepackage[utf8]{inputenc}
\usepackage[small]{caption}
\usepackage{graphicx}
\usepackage{xcolor}
\usepackage{amsmath}
\usepackage{booktabs}
\usepackage{algorithm}
\usepackage{dsfont}

\usepackage{amsthm}
\usepackage{amssymb}
\usepackage[noend]{algorithmic}
\usepackage{enumitem}
\usepackage{mathtools}

\usepackage{bbm}
\usepackage{multirow}
\usepackage{tikz}
\usetikzlibrary{calc,positioning}

\newtheorem{theorem}{Theorem}[section]
\newtheorem{corollary}[theorem]{Corollary}

\theoremstyle{definition}

\usepackage{natbib}
\usepackage{authblk}
\usepackage{cleveref}

\usepackage{aliascnt}

\newaliascnt{proposition}{theorem}
\newtheorem{proposition}[proposition]{Proposition}
\aliascntresetthe{proposition}
\crefname{proposition}{proposition}{propositions}
\Crefname{proposition}{Proposition}{Propositions}

\newcommand{\cA}{\mathcal{A}}
\newcommand{\cB}{\mathcal{B}}
\newcommand{\cC}{\mathcal{C}}
\newcommand{\cD}{\mathcal{D}}
\newcommand{\cP}{\mathcal{P}}

\allowdisplaybreaks

\title{\bf Reconfiguration in Fair Division Revisited}

\author[1]{Fabian Frank}
\author[2]{Warut Suksompong}

\affil[1]{Technical University of Munich, Germany}
\affil[2]{National University of Singapore, Singapore}

\date{\vspace{-10mm}}

\begin{document}

\maketitle

\begin{abstract}
We revisit reconfiguration in the fair allocation of indivisible goods, where the goal is to transform one fair allocation into another through a sequence of exchanges while preserving fairness at every step.
Our focus is on the hierarchy of \emph{envy-freeness up to $k$ goods (EF$k$)}.
We show that for any fixed~$k$, two EF1 allocations with the same size vector need not admit a reconfiguration path whose intermediate allocations satisfy EF$k$.
This impossibility persists even when the two allocations arise from standard EF1 approaches: the envy cycle elimination algorithm or the maximum Nash welfare solution.
In contrast, we prove that allocations with the same size vector produced by recursively balanced picking sequences, including round-robin, are always connected via a path that maintains EF2.
We also show that deciding whether an EF1 reconfiguration path exists is NP-hard for any fixed number of agents.
Furthermore, we complement these exchange-based results by studying a more permissive model that also allows transfers, establishing additional connectivity guarantees.
\end{abstract}

\section{Introduction}

Fair division refers to the study of how to allocate resources fairly among agents with potentially differing preferences.
In light of its wide applicability---ranging from cake cutting and dispute resolution~\citep{BramsTa96} to rent sharing and household chore allocation~\citep{GoldmanPr14}---it comes as no surprise that the topic has received substantial attention in mathematics, economics, and computer science.

Perhaps the most well-known fairness notion is \emph{envy-freeness}, which stipulates that no agent should envy another agent based on the chosen allocation \citep{Foley67,Varian74}.
However, in the ubiquitous setting of allocating indivisible goods such as artwork, furniture, or sports equipment, an envy-free allocation may fail to exist.
For instance, this occurs when all agents value a single common good more than the remaining goods combined.
As a consequence, researchers have introduced a relaxation called \emph{envy-freeness up to one good (EF1)} \citep{LiptonMaMo04,Budish11}.
An EF1 allocation always exists and can be found, for example, via the \emph{round-robin algorithm}, which lets the agents pick their favorite goods in cyclic order until the goods run out.
In fact, the same guarantee can be obtained by any \emph{recursively balanced picking sequence}, which lets each agent pick once in every round (except possibly the last round, where each agent picks at most once).
More generally, EF1 belongs to the hierarchy of \emph{envy-freeness up to $k$ goods (EF$k$)}, under which any envy can be eliminated by removing at most $k$~goods from the envied agent's bundle.

While much of the fair division literature concerns finding a fair allocation from scratch, recent work has increasingly addressed settings where an initial allocation is already given.
Inspired by the field of combinatorial reconfiguration \citep{Vandenheuvel13,Nishimura18}, \citet{IgarashiKaSu24} initiated the study of reconfiguration in fair division.
In this model, the goal is to transform one fair allocation into another through a sequence of local modifications while maintaining fairness at every intermediate step.
These authors focused on the fairness notion of EF1 and used (pairwise) exchanges as the reconfiguration operation.
Among other results, they showed that even when the initial and target allocations are both EF1, a reconfiguration path whose intermediate allocations all remain EF1 might not exist; however, such a path is guaranteed to exist for two agents with either identical or binary utilities.

In addition to their investigation, \citet[Sec.~5]{IgarashiKaSu24} also left open several important questions.
For example, since a reconfiguration path whose intermediate allocations are all EF1 may not exist, can existence be restored by allowing the intermediate allocations to be EF$k$ for some fixed $k > 1$?
Regarding complexity, while deciding the existence of an EF1 reconfiguration path is PSPACE-complete in general, can the problem be solved in polynomial time if there are only two (or a constant number of) agents?
And what happens if transfers are allowed in addition to exchanges?
In this paper, we provide answers to all of these questions as well as further results.

\subsection{Overview of Results}

We begin in \Cref{sec:existence} by considering reconfiguration paths between EF1 allocations where intermediate allocations are allowed to be EF$k$ for some $k \ge 1$.
We show that for any fixed number of agents $n\ge 2$ and any fixed~$k$, there exist initial and target EF1 allocations that cannot be connected via an EF$k$ reconfiguration path.
This impossibility persists even when both given allocations arise from standard EF1 approaches: the envy cycle elimination algorithm or, when there are at least three agents with identical utilities, the maximum Nash welfare solution.
On the other hand, when the given allocations are produced by recursively balanced picking sequences, we establish the existence of an EF2 reconfiguration path, thereby highlighting the robustness of allocations obtained from this method.
Moreover, if the number of goods is divisible by the number of agents, the two allocations can be connected by a path in which every allocation is also produced by a recursively balanced picking sequence (and therefore satisfies EF1); however, this ceases to hold without the divisibility condition.

Next, in \Cref{sec:complexity}, we investigate the complexity of deciding whether an EF1 reconfiguration path between two EF1 allocations exists, as well as the length of such paths.
We prove that the problem is NP-hard for every fixed number of agents $n\ge 2$, even when both agents receive the same number of goods in the initial allocation (and the target allocation).
We also give a polynomial-time reduction from the \textsc{Subset Sum Reconfiguration} problem to \textsc{EF1 Reconfiguration}---as a consequence, membership in NP for the latter would imply membership in NP for the former and resolve a question left open by \citet{ItoDe14}.
In addition, we demonstrate that even when an EF1 reconfiguration path for two agents exists, it is possible that every shortest such path has length quadratic in the number of goods, significantly exceeding the linear length of a shortest path when the EF1 constraint is ignored.

Finally, in \Cref{sec:transfers}, we examine a more permissive model that allows transfers in addition to exchanges.
For identical utilities, we establish full EF1 connectivity for any number of agents, thereby generalizing prior results of \citet[App.~A]{IgarashiKaSu24} which hold for either identical binary utilities or two agents with identical utilities.
For binary utilities, however, connectivity can still fail whenever there are at least five agents.
Nevertheless, we prove that connectivity can be recovered if we relax the intermediate fairness requirement to EF2---in fact, transfers alone already suffice for this guarantee.
We also show that for allocations produced by recursively balanced picking sequences, permitting transfers allows us to remove the divisibility restriction in the exchange-only setting: any two such allocations can be connected by allocations with the same property.

\subsection{Further Related Work}

As mentioned earlier, \citet{IgarashiKaSu24} initiated the study of reconfiguration in fair division, focusing primarily on EF1 and exchange operations while also considering a variant with transfers in the appendix.
\citet{Eyjolfsson26} followed up by allowing multiple goods to change owners in each operation and investigating envy-freeness in addition to EF1.
\citet{ChandramouleeswaranNiRa25} examined a model where, starting from a ``near-EF1'' allocation, the goal is to transfer goods in order to reach an EF1 allocation while traversing only near-EF1 allocations.
\citet{BentertBrDe25}, \citet{BredereckDeIn26}, \citet{YuenIgKa26}, \citet{BredereckKaLu23}, and \citet{DornDeSc21} along with \citet{BoehmerBrHe24} considered ``reformation'' models where an initial unfair allocation can be transformed into a fair allocation by adding, transferring, exchanging, sharing, and deleting goods, respectively.
Unlike in reconfiguration, in the model of \citet{ChandramouleeswaranNiRa25} as well as the reformation models, the objective is to reach any fair allocation rather than a given one.
Other related models include \emph{online fair division} \citep{AleksandrovWa20} and \emph{temporal fair division} \citep{CooksonEbSh25,ElkindLaLa25}, in which goods arrive one at a time---in an unknown order and a known order, respectively---and must be allocated immediately upon their arrival.
In contrast to reconfiguration, there is no initial allocation in these models.

Beyond fair division, reconfiguration has been studied for a range of combinatorial problems including Boolean satisfiability \citep{ItoDeHa11}, subset sum \citep{ItoDe14}, and graph coloring \citep{JohnsonKrKr16}.
More recently, within social choice, \citet{DongFrPe26} examined committee voting from a reconfiguration perspective, where each operation allows a committee to exchange one of its members with a non-member.

\section{Preliminaries}
\label{sec:prelims}

Let $N = \{1,\dots,n\}$ be a set of $n\ge 2$ agents, and $M$ be a set of $m\ge 1$ goods.
A \emph{bundle} refers to a (possibly empty) set of goods.
Each agent $i\in N$ has a utility function $u_i : 2^M \to \mathbb{R}_{\ge 0}$ that maps bundles to non-negative real numbers.
We write $u_i(g)$ instead of $u_i(\{g\})$ for a single good $g\in M$, and assume that the utility functions are additive, i.e., $u_i(M') = \sum_{g\in M'}u_i(g)$ for all $i\in N$ and $M'\subseteq M$.
The utility functions are called \emph{identical} if $u_i = u_j$ for all $i,j\in N$---we shall use $u$ to denote the common utility function in this case---and \emph{binary} if $u_i(g)\in\{0,1\}$ for all $i\in N$ and $g\in M$.
An \emph{instance} consists of $N$, $M$, and $(u_i)_{i\in N}$.

An \emph{allocation} $\cA = (A_1,\dots,A_n)$ is an ordered partition of~$M$ into $n$~bundles such that each bundle~$A_i$ is allocated to agent $i\in N$.
The \emph{size vector} of an allocation~$\cA$ is the vector $\vec{s} = (s_1,\dots,s_n)$ such that $s_i = |A_i|$ for each $i\in N$.
We call an allocation \emph{exactly-balanced}\footnote{The typical definition of \emph{balanced} allows agents to receive bundles of sizes differing by at most $1$ (see, e.g., \citep[Sec.~4]{Suksompong21}).} if $s_1 = \dots = s_n$.
An allocation~$\cA$ is \emph{envy-free} if for every pair of agents $i,j\in N$, it holds that $u_i(A_i) \ge u_i(A_j)$.
For a positive integer~$k$, an allocation~$\cA$ satisfies \emph{envy-freeness up to $k$ goods (EF$k$)} if for every pair of agents $i,j\in N$, there exists $B\subseteq A_j$ with $|B| \le k$ such that $u_i(A_i) \ge u_i(A_j\setminus B)$.
Two allocations $\cA$ and $\cB$ can reach each other via an exchange if there exist distinct agents $i,j\in N$ and goods $g\in A_i$ and $g'\in A_j$ such that $B_i = (A_i\cup\{g'\})\setminus\{g\}$, $B_j = (A_j\cup\{g\})\setminus\{g'\}$, and $B_\ell = A_\ell$ for all $\ell\in N\setminus\{i,j\}$.
An \emph{(exchange) reconfiguration path} between allocations $\cA$ and $\cB$ is a sequence of allocations $(\cA^0,\cA^1,\dots,\cA^T)$ for some $T\ge 0$ such that $\cA^0 = \cA$, $\cA^T = \cB$, and $\cA^{t-1}$ and $\cA^t$ can reach each other via an exchange for each $t \in \{1,\dots,T\}$.
We call $\cA$ the \emph{initial allocation} of the path and $\cB$ the \emph{target allocation}.
For a positive integer~$k$, if all allocations on a reconfiguration path satisfy EF$k$, we call the path an \emph{EF$k$ reconfiguration path}.
In \Cref{sec:transfers}, we allow transfers in addition to exchanges; we describe the alternative model in more detail at the beginning of that section.

\section{Reconfiguration Path Existence}
\label{sec:existence}

In this section, we consider the existence of EF$k$ reconfiguration paths when the initial and target allocations satisfy EF1 or are produced by established EF1 algorithms.

\citet{IgarashiKaSu24} proved that for any number of agents, an EF1 reconfiguration path between two given EF1 allocations with the same size vector may fail to exist.
Their proof relies on delicate constructions tailored to the requirement that intermediate allocations remain EF1.
We strengthen this result by showing that non-existence persists even when the intermediate allocations are allowed to satisfy EF$k$ for any fixed $k > 1$.

\begin{theorem}
\label{thm:general-nonexistence}
For any fixed $n \ge 2$ and $k \ge 1$, there exists an instance with $n$ agents together with initial and target EF1 allocations with the same size vector such that no EF$k$ reconfiguration path exists.    
\end{theorem}

\begin{proof}
We start by presenting an instance with $n = 2$ agents and later extend it to arbitrary $n$.
Fix $k \ge 1$ and let $r = 3k+2$.

Consider an instance with two agents where the set of goods consists of $X = \{x_1,\dots,x_{5r-3}\}$, $Y = \{y_1,\dots,y_{3r}\}$, and two special goods $z,w$; note that the total number of goods is $8r-1 = 24k+15$.
The agents' utilities for the goods are as follows.
\begin{center} 
    \begin{tabular}{c|cccc}
        & $x\in X$ & $y\in Y$ & $z$ & $w$\\
        \hline
        $u_1$ & $0$ & $1$ & $2r+1$ & $r$ \\
        $u_2$ & $1$ & $2$ & $r+1$ & $1$
    \end{tabular}
\end{center}
Let the initial allocation be $\cA = (A_1, A_2) = (X\cup \{z,w\}, Y)$, which has size vector $(5r-1, 3r)$.
For the target allocation, let $X_1\subseteq X$ and $Y_1\subseteq Y$ be such that $|X_1| = 4r-2$ and $|Y_1| = r$, and let $X_2 = X\setminus X_1$ and $Y_2 = Y\setminus Y_1$.
We have $|X_2| = r-1$ and $|Y_2| = 2r$.
Let the target allocation be $\cB = (B_1, B_2) = (X_1\cup Y_1\cup\{w\}, X_2\cup Y_2\cup \{z\})$, which also has size vector $(5r-1, 3r)$.

We first show that both $\cA$ and $\cB$ are EF1.
For $\cA$, we have $u_1(A_1) = (5r-3)\cdot 0 + (2r+1) + r = 3r+1$ and $u_1(A_2) = (3r)\cdot 1 = 3r$, so agent~$1$ does not envy agent~$2$.
Also, $u_2(A_2) = (3r)\cdot 2 = 6r$ and $u_2(A_1) = (5r-3)\cdot 1 + (r+1) + 1 = 6r-1$, so agent~$2$ does not envy agent~$1$.
This means that $\cA$ is envy-free and therefore EF1.
For $\cB$, we have $u_1(B_1) = (4r-2)\cdot 0 + r\cdot 1 + r$ and $u_1(B_2\setminus\{z\}) = u_1(X_2\cup Y_2) = (r-1)\cdot 0 + (2r)\cdot 1 = 2r$, so agent~$1$ does not envy agent~$2$ upon removing~$z$.
Also, $u_2(B_2) = (r-1)\cdot 1 + (2r)\cdot 2 + (r+1) = 6r$ and $u_2(B_1) = (4r-2)\cdot 1 + r\cdot 2 + 1 = 6r-1$, so agent~$2$ does not envy agent~$1$.
Therefore, $\cB$ is EF1 as well.

Assume for contradiction that there exists an EF$k$ reconfiguration path between $\cA$ and $\cB$.
Consider the first exchange that moves~$z$ from agent~$1$ to agent~$2$; such an exchange exists since $z\in A_1$ and $z\in B_2$. 
Let $t$ be the number of goods from~$Y$ held by agent~$1$ immediately before this exchange.
Since the size vector is always $(5r-1, 3r)$, at this point, agent~$1$'s bundle contains $z$ along with $t$~goods from~$Y$ and $5r-t-2$ goods from $X\cup\{w\}$, and so agent~$2$'s bundle contains $3r-t$ goods from~$Y$ along with $t$ goods from $X\cup\{w\}$.
Agent~$2$'s utility for her own bundle is therefore $(3r-t)\cdot 2 + t\cdot 1 = 6r-t$, whereas her utility for agent~$1$'s bundle is $(r+1) + t\cdot 2 + (5r-t-2)\cdot 1 = 6r+t-1$.
We claim that $2t \le r + 2k$.
This holds trivially if $t \le k$, so suppose that $t \ge k+1$.
Note that the $k$ most valuable goods in agent~$1$'s bundle for agent~$2$ are $z$ together with $k-1$ goods from~$Y$, and their total utility for agent~$2$ is $(r+1) + (k-1)\cdot 2 = r+2k-1$.
After removing them, agent~$2$'s utility for agent~$1$'s bundle is $(6r+t-1) - (r+2k-1) = 5r + t - 2k$.
Since this allocation must be EF$k$, it holds that $6r-t \ge 5r+t-2k$, or equivalently, $2t \le r + 2k$.
Hence, it always holds that $2t \le r + 2k$.

Next, let $t'$ be the number of goods from~$Y$ held by agent~$1$ immediately \emph{after} the first exchange that moves $z$ from agent~$1$ to agent~$2$.
We have $t' \le t + 1$, and so 
\begin{align}
\label{eq:tprime-upper}
2t' \le 2t + 2 \le r+2k+2 = 5k+4.
\end{align}
Since $z$ has moved to agent~$2$, agent~$1$'s utility for her own bundle is at most $t'\cdot 1 + r = t' + r$.
On the other hand, agent~$2$'s bundle contains $3r-t'$ goods from $Y$ along with $z$ (and possibly goods from $X\cup\{w\}$).
After removing the $k$ most valuable goods (for agent~$1$), agent~$1$'s utility for this bundle remains at least $(3r-t' - (k-1))\cdot 1 = 3r-t'-k + 1$.
Since this allocation must be EF$k$, it holds that $t' + r \ge 3r-t'-k+1$, or equivalently, $2t' \ge 2r-k+1 = 5k+5$, contradicting~\eqref{eq:tprime-upper}.
It follows that no EF$k$ reconfiguration path exists between $\cA$ and $\cB$.

\medskip

We now describe how to extend this instance to arbitrary $n > 2$.
Let $m_0 = 8r-1 = 24k+15$ be the number of goods in the original instance.
For each additional agent $i\in\{3,\dots,n\}$, we introduce one new good $h_i$.
Both original agents have utility~$0$ for all goods~$h_i$.
Each new agent~$i$ has utility~$1$ for each original good, utility $m_0$ for~$h_i$, and utility~$0$ for all remaining goods $h_j$ with $j \ne i$.
We extend the initial allocation~$\cA$ and the target allocation~$\cB$ by giving each new good~$h_i$ to agent~$i$.
Note that the size vector of $\cA$ and $\cB$ is $(5r-1, 3r, 1, \dots, 1)$, and the total number of goods is $m_0 + (n-2) = 24k + n + 13$.
Moreover, both $\cA$ and $\cB$ are EF1.
Indeed, agents~$1$ and $2$ are EF1 toward each other in the original instance and have utility~$0$ for the additional goods, while each new agent has utility~$m_0$ for her own bundle and at most~$m_0$ for any other bundle.

It remains to show that, again, there does not exist an EF$k$ reconfiguration path between $\cA$ and $\cB$.
Assume for contradiction that such a path exists.
We first claim that no new agent can participate in an exchange.
Otherwise, consider the first exchange involving a new agent.
Before this exchange, every new agent still owns her high-value good, and all original goods remain divided between the two original agents.
If two new agents exchange goods with each other, each of their utilities drops to~$0$, and their utility for an original agent's bundle after removing $k$ goods is at least $3r - k = 8k+6 > 0$, so the allocation is not EF$k$.
Else, suppose that a new agent exchanges goods with an original agent~$i$.
In this case, the new agent's utility drops to~$1$, and her utility for agent~$i$'s bundle after removing the $k$ most valuable goods is at least $(3r-1) - (k-1) = 3r-k = 8k+6 > 1$, so the allocation again violates EF$k$.
Thus, no new agent can be involved in an exchange, that is, every exchange must be between the two original agents.
However, such an EF$k$ reconfiguration path would give rise to an EF$k$ reconfiguration path in the original $n = 2$ instance, a contradiction.
It follows that no EF$k$ reconfiguration path between $\cA$ and $\cB$ exists, completing the proof.
\end{proof}

In light of \Cref{thm:general-nonexistence}, a natural question is whether existence can be recovered if the initial and target allocations are not arbitrary EF1 allocations, but instead arise from standard EF1 algorithms.
We shall investigate this question for three well-known approaches to obtaining EF1 allocations.

The first approach is based on \emph{recursively balanced picking sequences}.
In a picking sequence, agents take turns selecting their favorite good among the remaining goods (with ties broken arbitrarily) according to a predetermined order.
Such a sequence can be divided into ``rounds'' of $n$~picks each, except possibly the final round, which contains fewer than $n$~picks if $m$ is not divisible by~$n$.
A~picking sequence is called recursively balanced if each agent picks at most once in every round; these picking sequences are precisely the ones guaranteeing EF1 \citep[Prop.~2.1]{CelineSuYu26}.
For brevity, we call an allocation that can be produced by some recursively balanced picking sequence a \emph{recursively balanced allocation}.
Interestingly, we show that if the initial and target allocations are recursively balanced (and have the same size vector), then there exists an EF2 reconfiguration path.
This positive result underscores the robustness of recursively balanced allocations from a reconfiguration perspective.

\begin{theorem}
\label{thm:recursively-balanced-general}
For any instance, if the initial and target allocations are recursively balanced allocations and have the same size vector, then there exists an EF2 reconfiguration path.
\end{theorem}

\begin{proof}
Fix an instance.
First, we claim that any allocation produced by a recursively balanced picking sequence~$\pi$ with at most one ``forced pick''---we refer to such an allocation as an ``almost recursively balanced allocation''---satisfies EF2.
In a forced pick, the picking agent is forced to pick a certain good which may not be her most valuable good among the remaining goods.
To prove this claim, consider any two agents $i$ and $j$, and let $t$ be the number of goods that $j$ picks.
Since $\pi$ is recursively balanced, for each $\ell\in\{2,\dots,t\}$, agent~$i$'s $(\ell-1)$st pick comes before agent~$j$'s $\ell$th pick. 
Hence, agent~$i$ values her good from the former pick at least as much as agent~$j$'s good from the latter pick; the only possible exception is when $i$'s former pick is forced.
This means that if we remove the good that $j$ picks in the round after $i$'s forced pick (if this good exists) as well as the good that $j$ picks in the first round, $i$ does not envy~$j$.
Since this holds for every pair of agents $i$ and $j$, the allocation satisfies EF2.

Next, for each agent who has the same utility for multiple goods, we temporarily fix a tie-breaking order for the agent among such goods---this makes the output of any given picking sequence unique.
Consider any suffix $(c_1,c_2,\dots,c_\ell)$ of a picking sequence~$\pi$, and let $M_0\subseteq M$ be the set of $\ell$ remaining goods.
For each $g\in M_0$, let $\cA(g)$ be the allocation produced by forcing agent~$c_1$ to pick~$g$ in the first pick of this suffix, while letting agents pick their most preferred good (subject to the chosen tie-breaking) in all other picks (including those before the suffix).

We claim that for any goods $g, g'\in M_0$, there exists a reconfiguration path between $\cA(g)$ and $\cA(g')$ such that every intermediate allocation corresponds to $\cA(h)$ for some $h\in M_0$.
We prove this claim by induction on~$\ell$, the length of the suffix.
For the base case $\ell = 1$, the claim holds trivially.
For the inductive step, suppose that the claim holds when the suffix has length $\ell-1$, and consider a suffix of length~$\ell \ge 2$.
Let agent~$c_2$'s two most preferred goods in~$M_0$ be $x$ and $y$, respectively.
Whenever $g \ne x$, after agent~$c_1$ is forced to take~$g$, agent~$c_2$ takes~$x$.
If we remove $x$ along with agent~$c_2$'s pick in~$\pi$, the suffix has length $\ell - 1$.
For $g\in M_0\setminus\{x\}$, define $\cB(g)$ for the reduced picking sequence in the same way as $\cA(g)$ for the original sequence.
By the inductive hypothesis, for any $g,g'\in M_0\setminus\{x\}$, there exists a reconfiguration path between $\cB(g)$ and $\cB(g')$ such that every intermediate allocation corresponds to $\cB(h)$ for some $h\in M_0\setminus\{x\}$.
Hence, the same is true for $\cA(g)$ and $\cA(g')$ for all $g,g'\in M_0\setminus\{x\}$.
Now, consider $\cA(x)$ and $\cA(y)$.
In $\cA(x)$, agent~$c_1$ picks~$x$ and agent~$c_2$ picks~$y$, whereas in $\cA(y)$, agent~$c_1$ picks~$y$ and agent~$c_2$ picks~$x$; the subsequent picks are the same in the two allocations.
Therefore, $\cA(x)$ and $\cA(y)$ can reach each other via an exchange if agents $c_1$ and $c_2$ are different, and coincide if agents $c_1$ and $c_2$ are the same.
Since $y \in M_0\setminus\{x\}$, it follows that for any $g,g'\in M_0$, there exists a reconfiguration path between $\cA(g)$ and $\cA(g')$ such that every intermediate allocation corresponds to $\cA(h)$ for some $h\in M_0$, completing the induction.

For our next step, with the agents' tie-breaking orders still fixed, we consider any two recursively balanced picking sequences $\pi,\pi'$ that differ only by swapping two adjacent agents $i,j$.
Assume without loss of generality that $i$ appears before $j$ in $\pi$, and $j$ appears before $i$ in $\pi'$.
Consider the suffix of~$\pi$ beginning with agent~$i$'s pick, and define $\cA(g)$ as before for each remaining good~$g$ with respect to this suffix.
In a normal run of~$\pi'$, suppose that agent~$j$ chooses $h$ and agent~$i$ chooses~$h'$.
In the run of~$\pi$, if we force $i$ to choose $h'$, then $h$ remains $j$'s most preferred good among the remaining goods, and so $j$ chooses $h$.
Therefore, the resulting allocation corresponds to $\cA(h')$.
Let $\widehat{h}$ be the good that $i$ picks in a normal run of~$\pi$.
By our earlier claims, there exists a reconfiguration path between $\cA(h')$ and $\cA(\widehat{h})$ such that every intermediate allocation corresponds to $\cA(h'')$ for some $h''$ and is therefore EF2.
This implies that the allocations produced by $\pi$ and $\pi'$ are connected by an EF2 reconfiguration path.

We now lift the assumption of fixed tie-breaking orders.
Suppose that an agent~$i$ reverses the order of two adjacent goods in her tie-breaking order, from $g,g'$ to $g',g$.
Consider the run of an arbitrary recursively balanced picking sequence.
If the two tie-breaking orders yield different allocations, consider the first turn at which they differ.
In this turn, agent~$i$ picks $g$ for one tie-breaking order and $g'$ for the other.
Viewing the latter choice as forcing $i$ to pick~$g'$ despite the order being the former, we again get that these two allocations are connected by an EF2 reconfiguration path.
Since any two tie-breaking orders can be transformed into each other by repeatedly swapping the order of adjacent goods, the resulting allocations are connected by an EF2 reconfiguration path.
Applying this argument across all agents, it follows that for any two recursively balanced picking sequences that differ only by swapping two adjacent agents, the resulting allocations are connected by an EF2 reconfiguration path.
In particular, the two picking sequences need not rely on the same tie-breaking order.

Finally, consider any two recursively balanced allocations with the same size vector resulting from picking sequences $\pi$ and $\pi'$.
Since both allocations have the same size vector, the final (possibly incomplete) round of $\pi$ and $\pi'$ consists of the same set of agents.
Therefore, $\pi$ can be transformed into~$\pi'$ by repeatedly swapping adjacent agents in such a way that every intermediate picking sequence remains recursively balanced.
By the previous paragraph, the two allocations are connected by an EF2 reconfiguration path, completing the proof.
\end{proof}

If the number of goods is a multiple of the number of agents, we show that the EF2 guarantee can be improved to EF1.
In fact, we can even ensure that every intermediate allocation is recursively balanced.

\begin{theorem}
\label{thm:recursively-balanced-divisible}
For any instance where $m$ is divisible by $n$, and any initial and target recursively balanced allocations, there exists a reconfiguration path such that every intermediate allocation is recursively balanced (and therefore EF1). 
\end{theorem}

\begin{proof}
Consider any instance where $m$ is divisible by~$n$.
Let $M'\subseteq M$ be a nonempty set of goods, and suppose that the goods in $M\setminus M'$ have been allocated in a fixed manner.
Let $\emptyset \neq N'\subseteq N$ be such that $|M'| = |N'| + rn$ for some integer $r\ge 0$.
Consider all picking sequences of the form $(\sigma_0 \mid \sigma_1 \mid \dots \mid \sigma_r)$, where $\sigma_0$ is an arbitrary permutation of~$N'$ and $\sigma_1,\dots,\sigma_r$ are arbitrary permutations of~$N$.
Let $\Gamma$ be the set of all allocations that can be produced by any such sequence (with ties broken arbitrarily).
We claim that for any two allocations in~$\Gamma$, there exists a reconfiguration path such that every intermediate allocation also belongs to~$\Gamma$.
This is sufficient to establish the theorem, as we may take $M' = M$.

We prove this claim by induction on $|M'|$.
For the base case $|M'| = 1$, the claim holds trivially.
For the inductive step, suppose that $|M'| \ge 2$, and assume that the claim holds whenever this size is smaller by~$1$.
For $i\in N'$ and $g\in M'$, let $\Gamma_{i,g}$ be the set of all allocations that can be produced if, in the first pick of~$\sigma_0$, agent~$i$ picks good~$g$.
By the inductive hypothesis, for any two allocations in~$\Gamma_{i,g}$, there exists a reconfiguration path such that every intermediate allocation also belongs to~$\Gamma_{i,g}$.

Take any distinct agents $i,j\in N'$ (assuming $|N'| \ge 2$), and let $g$ and $h$ be most valuable goods in~$M'$ for $i$ and $j$, respectively, with ties broken arbitrarily.
We will connect $\Gamma_{i,g}$ with $\Gamma_{j,h}$.
Assume first that $g\ne h$.
In this case, consider a run beginning with $i$ picking~$g$ and then $j$ picking~$h$, and another run beginning with $j$ picking~$h$ and then $i$ picking~$g$.
After the first two picks, both runs have assigned $g$ to~$i$ and $h$ to~$j$, and the set of remaining goods is identical.
Thus, the remaining picks can be completed in the same way, including with the same tie-breaking (if any).
This means that $\Gamma_{i,g} \cap \Gamma_{j,h} \neq \emptyset$.

Assume now that $g = h$.
We will construct two runs.
In $\sigma_0$ of the first (resp., second) run, let agent~$i$ (resp., agent~$j$) pick first, agent~$j$ (resp., agent~$i$) pick last, and agents in $N'\setminus\{i,j\}$ pick in between in the same order.
In both runs, the first agent can pick~$g$ and the agents in $N'\setminus\{i,j\}$ can pick identically.
If the last agent also picks the same good~$g'$ in both runs, we can complete the remainder of the runs identically, and the resulting allocations can reach each other by exchanging $g$ and~$g'$.
Else, assume that the last agent~$j$ picks $g_0$ in the first run, and the last agent~$i$ picks $h_0\ne g_0$ in the second run.
Since the set of remaining goods after $\sigma_0$ differs in the two runs, the next (complete) round $\sigma_1$ exists.

We repeat a similar process: in $\sigma_1$ of the first (resp., second) run, let agent~$i$ (resp., agent~$j$) pick first, agent~$j$ (resp., agent~$i$) pick last, and agents in $N\setminus\{i,j\}$ pick in between in the same order.
In the first run, agent~$i$ can pick~$h_0$, and in the second run, agent~$j$ can pick~$g_0$.
The set of remaining goods is the same in both runs, so the agents in $N\setminus\{i,j\}$ can pick identically.
We then check whether the good~$g_1$ that the last agent~$j$ picks in the first run coincides with the good $h_1$ that the last agent~$i$ picks in the second run.
We proceed in this manner until at the end of some round, the last agent picks the same good in the two runs---such a round must exist because immediately before the last pick of the final (complete) round, exactly one good remains.
Assume that at the end of this round, agent~$j$ picks good~$g_t$ in the first run, and agent~$i$ also picks good~$g_t$ in the second run.
After this pick, the (possibly empty) set of remaining goods is identical in both runs, and the remaining picks can be completed in the same way.
The only difference in the resulting allocations is that in the first run, agent~$i$ owns $g$ and agent~$j$ owns $g_t$, while this is reversed in the second run.
In other words, these allocations can reach each other by exchanging $g$ and~$g_t$.
That is, an allocation in~$\Gamma_{i,g}$ can reach an allocation in $\Gamma_{j,h} = \Gamma_{j,g}$ via an exchange.

It remains to connect $\Gamma_{i,g}$ and $\Gamma_{i,g'}$ for $i\in N'$ and $g\ne g'$, where both $g$ and $g'$ are most valuable goods of agent~$i$ in~$M'$.
Assume first that $|N'| \ge 2$.
Choose any $j\in N'\setminus\{i\}$, and let $h$ be a most valuable good of agent~$j$ in~$M'$ (possibly $h\in\{g,g'\}$).
Our previous argument connects $\Gamma_{i,g}$ with $\Gamma_{j,h}$, and $\Gamma_{j,h}$ with $\Gamma_{i,g'}$, which means that $\Gamma_{i,g}$ and $\Gamma_{i,g'}$ are also connected.
Finally, assume that $|N'| = 1$.
Since $|M'| \ge 2$, the (complete) round $\sigma_1$ exists.
Consider the following two runs.
In the first run, agent~$i$ takes~$g$ in her first pick, is placed first in the next round, and takes~$g'$ in that pick.
In the second run, agent~$i$ takes~$g'$ in her first pick, is placed first in the next round, and takes~$g$ in that pick.
The remainder of both runs can be completed in the same way, and the resulting allocations are the same, so $\Gamma_{i,g} \cap \Gamma_{i,g'} \neq \emptyset$.
Hence, for any two allocations in~$\Gamma$, there exists a reconfiguration path such that every intermediate allocation also belongs to~$\Gamma$.
This completes the induction and therefore the proof.
\end{proof}

The condition that $m$ is divisible by~$n$ is crucial for \Cref{thm:recursively-balanced-divisible}.
As the next proposition demonstrates, without this condition, one cannot guarantee that the intermediate allocations are recursively balanced.

\begin{proposition}
\label{prop:recursively-balanced-non-divisible}
There exists an instance together with initial and target recursively balanced allocations with the same size vector for which no reconfiguration path such that every intermediate allocation is also recursively balanced exists.
\end{proposition}

\begin{proof}
Consider an instance with $n = 3$ agents and $m = 4$ goods $g_1,\dots,g_4$.
The agents' utility functions are such that
\begin{itemize}
\item $u_1(g_1) > u_1(g_2) > u_1(g_3) > u_1(g_4)$;
\item $u_2(g_4) > u_2(g_1) > u_2(g_3) > u_2(g_2)$;
\item $u_3(g_4) > u_3(g_2) > u_3(g_1) > u_3(g_3)$.
\end{itemize}
Six recursively balanced picking sequences produce an allocation with size vector $(2,1,1)$.
The picking sequences $(1,2,3,1)$, $(2,1,3,1)$, and $(2,3,1,1)$ produce the allocation $\cA^1 = (\{g_1,g_3\},\{g_4\},\{g_2\})$, the picking sequences $(1,3,2,1)$ and $(3,1,2,1)$ produce the allocation $\cA^2 = (\{g_1,g_2\},\{g_3\},\{g_4\})$, and the picking sequence $(3,2,1,1)$ produces the allocation $\cA^3 = (\{g_2,g_3\}, \{g_1\}, \{g_4\})$.

Observe that $\cA^1$ differs from $\cA^2$ in the bundles of all three agents.
Since a single exchange can alter the bundles of only two agents, $\cA^1$ cannot reach $\cA^2$ via an exchange.
Similarly, $\cA^1$ cannot reach~$\cA^3$ via an exchange.
It follows that there is no reconfiguration path between $\cA^1$ and $\cA^2$ such that every intermediate allocation is also recursively balanced.
\end{proof}

It remains open whether \Cref{prop:recursively-balanced-non-divisible} holds if we only require the intermediate allocations to be EF1 instead of recursively balanced.
In particular, $\cA^1$ and $\cA^2$ can be connected via the intermediate allocations $(\{g_2,g_3\}, \{g_4\}, \{g_1\})$ and $\cA^3$, which are EF1 for any utilities consistent with the specified rankings.

The second approach that we investigate is the \emph{envy cycle elimination (ECE)} algorithm of \citet{LiptonMaMo04}.
This algorithm works by allocating one good at a time in arbitrary order.
Each good is allocated to an unenvied agent, and any resulting envy cycle is resolved by giving the bundle of each agent in the cycle to the agent who envies her.
This ensures the existence of an unenvied agent before the next good is allocated.
We show that in our general EF1 counterexample (\Cref{thm:general-nonexistence}), the initial and target allocations can also be produced by the ECE algorithm, which means that the same negative result holds for ECE.

\begin{theorem}
\label{thm:ECE-nonexistence}
For any fixed $n \ge 2$ and $k \ge 1$, there exists an instance with $n$ agents together with initial and target allocations that can be obtained by the ECE algorithm and have the same size vector such that no EF$k$ reconfiguration path exists.    
\end{theorem}

\begin{proof}
It suffices to show that in the examples constructed in \Cref{thm:general-nonexistence}, both the initial and target allocations can be obtained by the ECE algorithm.

We start with the example for $n = 2$.
For the allocation~$\cA = (X\cup\{z,w\}, Y)$, we perform the following steps.
\begin{enumerate}[label=\arabic*.]
\item Allocate one good from~$Y$ to agent~$1$.
\item Repeat the following $2r$ times: Allocate two goods from $X$ to agent~$2$, followed by one good from~$Y$ to agent~$1$.
\item Allocate $z$ to agent~$2$.
\item Allocate the remaining $r-3$ goods from~$X$ to agent~$2$.
\item Allocate the remaining $r-1$ goods from~$Y$ to agent~$1$.
\item Allocate $w$ to agent~$2$.
\end{enumerate}
We claim that an envy cycle occurs only at the end of this process---in particular, agent~$1$ does not envy agent~$2$ before the end of the process---and each good is allocated to an unenvied agent.
During Steps 1 and 2, agent~$2$ only receives goods from~$X$, which agent~$1$ values at~$0$.
Also, for $t\in\{0,\dots,2r-1\}$, after $t$ iterations of Step~2, agent~$1$ holds $t+1$ goods from~$Y$ and agent~$2$ holds $2t$ goods from~$X$.
Once the next two goods from~$X$ are allocated to agent~$2$, her utility for her own bundle becomes $2t+2$, which is the same as her utility for agent~$1$'s bundle---this means that agent~$1$ is unenvied before the next good from~$Y$ is allocated.
After Step~3, agent~$1$'s utility becomes $2r+1$ for each of the two bundles, and this remains true during Step~4.
At the beginning of Step~5, agent~$2$'s utility for her own bundle is $6r-2$ and her utility for agent~$1$'s bundle is $4r+2$; the latter utility increases to $6r-2$ before the last good of this step is allocated, so agent~$1$ is unenvied.
At the end of Step~5, agent~$1$ has utility $3r$ for her own bundle and $2r+1$ for agent~$2$'s bundle, so agent~$2$ is unenvied.
At the end of Step~6, the allocation is $(Y, X\cup\{z,w\})$.
Agent~$1$ has utility $3r$ for her bundle and $3r+1$ for agent~$2$'s bundle, whereas agent~$2$ has utility $6r-1$ for her bundle and $6r$ for agent~$1$'s bundle.
Thus, there is an envy cycle between the two agents.
Resolving the envy cycle results in the allocation~$\cA$.

Next, we consider the allocation $\cB = (X_1\cup Y_1\cup\{w\}, X_2\cup Y_2\cup \{z\})$.
Partition $Y_2$ into two sets $Y'_2, Y''_2$, each of size~$r$.
Allocate $w, X_2, Y'_2, Y_1, Y''_2, z, X_1$ to agent $1, 2, 2, 1, 2, 2, 1$ in this order.
We claim that no envy cycle occurs during this process and each good is allocated to an unenvied agent.
Initially, agent~$1$ receives utility~$r$ for~$w$.
After agent~$2$ receives $X_2$ and $Y'_2$, agent~$1$'s utility for agent~$2$'s bundle becomes~$r$, so agent~$1$ does not envy agent~$2$ up to this point. 
Agent~$2$ has utility $3r-1$ for her own bundle and $1$ for agent~$1$'s bundle.
When agent~$1$ next receives $Y_1$, her utility for her own bundle becomes $2r$ while agent~$2$'s utility for agent~$1$'s bundle increases to $1 + r\cdot 2 = 2r+1$.
After agent~$2$ receives $Y''_2$, agent~$1$'s utility for agent~$2$'s bundle also becomes $2r$, equal to agent~$1$'s utility for her own bundle.
At this point, agent~$2$'s utility for her own bundle is $5r-1$.
After receiving~$z$, this becomes~$6r$, while even after agent~$1$ receives~$X_1$, agent~$2$'s utility for agent~$1$'s bundle is only $6r-1$.
Thus, no envy cycle occurs and each good is allocated to an unenvied agent, and the process produces the allocation~$\cB$.
This completes the proof for the $n = 2$ example.

\medskip

Finally, we consider the example for general~$n$.
First, for $j\in\{3,\dots,n\}$, we allocate $h_j$ to agent~$j$ in arbitrary order.
Since each good~$h_j$ is positively valued only by agent~$j$, there is no envy cycle.
We then allocate the goods that are ultimately held by agents $1$ and~$2$ as in the $n = 2$ example.
Clearly, agents~$1$ and $2$ never envy the additional agents.
Moreover, each additional agent has utility $m_0$ for her own bundle and $m_0$ for the remaining goods combined, so she never envies any agent.
Therefore, the only possible envy cycle is between agents~$1$ and~$2$.
It follows that the same allocation process as in the $n = 2$ example can be performed in order to arrive at the desired allocations $\cA$ and $\cB$.
\end{proof}

Another prominent method for producing EF1 allocations is the \emph{maximum Nash welfare (MNW)} solution \citep{CaragiannisKuMo19}.
The \emph{Nash welfare} of an allocation is defined as the product of agents' utilities from the allocation, and an allocation is called an \emph{MNW allocation} if it maximizes the Nash welfare among all allocations.\footnote{Careful tie-breaking is necessary if the maximum possible Nash welfare is $0$, but this will not be relevant for our purposes.}
We show that whenever there are at least three agents, even if the initial and target allocations are MNW allocations which are moreover envy-free, and the agents have identical utilities, an EF$k$ reconfiguration path may still fail to exist.
This strengthens the construction of \citet[Thm.~4.9]{IgarashiKaSu24} for an EF1 reconfiguration path and identical utilities, where the initial and target allocations are EF1 but neither MNW nor envy-free.

\begin{theorem}
\label{thm:MNW-nonexistence}
For any fixed $n \ge 3$ and $k \ge 1$, there exists an instance with $n$ agents having identical utilities together with initial and target MNW allocations which are envy-free and have the same size vector, such that no EF$k$ reconfiguration path exists.    
\end{theorem}

\begin{proof}
Fix $n \ge 3$ and $k \ge 1$, and let $r = 4k+6$. 
Consider an instance with $n$~agents where the set of goods consists of $X = \{x_1,x_2\}$, $Y = \{y_1,y_2\}$, and $Z_i = \{z_{i,1},\dots,z_{i,r}\}$ for each $i\in \{3,\dots,n\}$; note that the number of goods is $(n-2)r + 4$. 
Every agent has the same utility function~$u$ such that $u(x_1) = 3k+4$, $u(x_2) = k+2$, $u(y_1) = u(y_2) = 2k+3$, and $u(z_{i,\ell}) = 1$ for all $i\in\{3,\dots,n\}$ and $\ell\in\{1,\dots,r\}$.
Observe that each agent's utility for each of $X$, $Y$, and $Z_i$ for $i\in\{3,\dots,n\}$ is $4k+6 = r$.

Let $\cA = (X, Y, Z_3, \dots, Z_n)$ and $\cB = (Y, X, Z_3, \dots, Z_n)$.
Since each agent has utility $r$ for every bundle in both $\cA$ and $\cB$, these two allocations are envy-free.
Moreover, each agent's total utility for all goods is $nr$.
Since the agents have identical utilities, by the inequality of arithmetic and geometric means, the Nash welfare is maximized when each agent receives utility exactly~$r$.
Hence, both $\cA$ and $\cB$ are MNW allocations.

Assume for contradiction that an EF$k$ reconfiguration path from $\cA$ to $\cB$ exists.
Since $A_1 \ne B_1$ and $A_2 \ne B_2$, some exchange must involve agent~$1$ or $2$.
Consider the first exchange involving either of these two agents.
Before this exchange, agent~$1$ has $X$, agent~$2$ has~$Y$, and each remaining agent has $r$ of the remaining goods.
If the exchange is between agents $1$ and $2$, then after the exchange, one of these agents has utility at most $3k+5$.
However, even after removing $k$ goods from any remaining agent's bundle, the utility for this bundle remains $r-k = 3k+6$.
This means that the allocation after the exchange is not EF$k$, a contradiction.
Hence, the exchange must be between an agent $i\in\{1,2\}$ and an agent $j\in N\setminus\{1,2\}$.
After the exchange, agent~$i$ has utility at most $r - (k+2) + 1 = 3k+5$.
On the other hand, after removing any $k$ goods from agent~$j$'s bundle, the utility for this bundle remains at least $r-k = 3k+6$.
Again, this means that the allocation after the exchange is not EF$k$, yielding the desired contradiction.
\end{proof}

\section{Complexity and Path Length}
\label{sec:complexity}

In this section, we examine the complexity of deciding whether an EF1 reconfiguration path exists between two EF1 allocations, as well as the worst-case minimum length of such a path when it exists.
We also discuss the relationship between these two problems.

For convenience, let \textsc{EF1 Reconfiguration} refer to the following problem: given an instance along with initial and target EF1 allocations, determine whether there exists an EF1 reconfiguration path between the two allocations.
\citet{IgarashiKaSu24} showed that \textsc{EF1 Reconfiguration} is PSPACE-complete for a general number of agents~$n$, but left open whether the problem can be solved in polynomial time for two agents.
We demonstrate that this is unlikely, by proving that the problem is NP-hard for any constant~$n$.

\begin{theorem}
\label{thm:NP-hard}
For any fixed number of agents $n \ge 2$, \textsc{EF1 Reconfiguration} is NP-hard, even when the initial and target allocations are exactly-balanced.
\end{theorem}

\begin{proof}
We start by presenting the reduction for $n = 2$ agents and later extend it to arbitrary~$n$.

For $n = 2$, we reduce from the NP-hard problem \textsc{Partition}.
In this problem, we are given positive integers $a_1,\dots,a_r$ with sum $2t$ for some positive integer~$t$, and the task is to determine whether these integers can be partitioned into two parts such that the sum of each part is exactly~$t$.
Let $T = 10t$.
Given an instance of \textsc{Partition}, we create an instance of \textsc{EF1 Reconfiguration} as follows.
The set of goods~$M$ consists of $X = \{x_1,x_2\}$, $Y = \{y_1,y_2,y_3,y_4\}$, $Z = \{z_1,z_2\}$, $G = \{g_1,\dots,g_r\}$, and $H = \{h_1,\dots,h_{r+2}\}$; note that the total number of goods is $2r+10$.
The utilities for the goods in $X$ and $Y$ are shown in the following table.
\begin{center} 
    \begin{tabular}{c|cccccc}
        & $x_1$ & $x_2$ & $y_1$ & $y_2$ & $y_3$ & $y_4$ \\
        \hline
        $u_1$ & $12T+t$ & $12T$ & $8T$ & $8T$ & $8T$ & $8T$ \\
        $u_2$ & $12T$ & $12T$ & $3T+t$ & $3T$ & $3T$ & $3T$
    \end{tabular}
\end{center}
Both agents value the remaining goods identically: for $i\in\{1,2\}$, $u_i(z_1) = 2t$, $u_i(z_2) = 3t$, $u_i(g_\ell) = a_\ell$ for $\ell\in\{1,\dots,r\}$, and $u_i(h_\ell) = 0$ for $\ell\in\{1,\dots,r+2\}$.
The initial allocation is
\[
\cA = 
(\{x_1,x_2,z_1,g_1,\dots,g_r,h_1,h_2\},\{y_1,y_2,y_3,y_4,z_2,h_3,\dots,h_{r+2}\})
\]
and the target allocation is
\[
\cB =
(\{x_1,x_2,z_2,h_1,\dots,h_{r+2}\}, \{y_1,y_2,y_3,y_4,z_1,g_1,\dots,g_r\}).
\]
Clearly, this instance can be constructed in polynomial time.
Since each agent receives $r+5$ goods in both allocations, these allocations are exactly-balanced.
Moreover, one can check that 
\[
u_1(A_1) = 24T + 5t >  24T + 3t = u_1(A_2\setminus\{y_1\})
\]
and
\[
u_2(A_2) = 12T + 4t = u_2(A_1\setminus\{x_1\}), 
\]
so $\cA$ is EF1. Similarly,
\[
u_1(B_1) = 24T + 4t = u_1(B_2\setminus\{y_1\})
\]
and
\[
u_2(B_2) = 12T + 5t > 12T + 3t = u_2(B_1\setminus\{x_1\}),
\]
so $\cB$ is EF1 as well.

We now establish the correctness of the reduction by proving that a Yes-instance of \textsc{Partition} corresponds to a Yes-instance of \textsc{EF1 Reconfiguration}, and vice versa. 

\medskip

($\Leftarrow$) Suppose that the \textsc{EF1 Reconfiguration} instance is a Yes-instance.
Consider any EF1 reconfiguration path.

First, we claim that the goods in $X$ and $Y$ cannot move.
Assume otherwise, and consider the first exchange that involves a good in either $X$ or~$Y$.
If it involves a good in~$X$, then after the exchange, agent~$1$'s utility for her own bundle is at most $20T+8t < 21T$, while her utility for agent~$2$'s bundle even after removing any good is still at least $24T$.
This means that EF1 is violated for agent~$1$, a contradiction.
Else, the exchange involves a good in~$Y$ but neither of the goods in~$X$.
After the exchange, agent~$2$'s utility for her own bundle is at most $9T+8t < 10T$, while her utility for agent~$1$'s bundle even after removing any good is still at least $15T$.
This means that EF1 is violated for agent~$2$, again a contradiction.
Therefore, the goods in $X$ and $Y$ cannot move.
In particular, agent~$1$'s highest utility for a good in agent~$2$'s bundle is always $8T$, and agent~$2$'s highest utility for a good in agent~$1$'s bundle is always $12T$.

For each agent $i\in \{1,2\}$ and each allocation, let $\delta_i$ denote agent~$i$'s utility for the goods in $M\setminus(X\cup Y)$ held by herself, and let $\Delta = \delta_1 - \delta_2$.
Observe that agent~$1$ is EF1 exactly when $24T + t + \delta_1 \ge (32T + \delta_2) - 8T$, that is, $\Delta \ge -t$.
Similarly, agent~$2$ is EF1 exactly when $12T + t + \delta_2 \ge (24T + \delta_1) - 12T$, that is, $\Delta \le t$.
Hence, every allocation on the reconfiguration path must satisfy $-t\le \Delta \le t$.

Next, for any allocation, let $\beta\in [0,2t]$ be the total utility of goods in~$G$ held by agent~$1$.
If both $z_1$ and $z_2$ are with agent~$1$, we have $\Delta = (5t+\beta) - (2t-\beta) = 3t + 2\beta > t$, so the allocation cannot be EF1.
Likewise, if both $z_1$ and $z_2$ are with agent~$2$, we have $\Delta = \beta - (7t-\beta) = 2\beta - 7t < -t$, so the allocation cannot be EF1.
Thus, for every allocation on the reconfiguration path, one agent holds~$z_1$ and the other agent holds~$z_2$.
When agent~$1$ holds~$z_1$ and agent~$2$ holds~$z_2$, we have $\Delta = (2t+\beta) - (5t-\beta) = 2\beta - 3t$, so the allocation is EF1 if and only if $t \le \beta \le 2t$.
When agent~$2$ holds~$z_1$ and agent~$1$ holds~$z_2$, we have $\Delta = (3t+\beta) - (4t-\beta) = 2\beta - t$, so the allocation is EF1 if and only if $0 \le \beta \le t$.
Since $z_1$ is with agent~$1$ in~$\cA$ and with agent~$2$ in $\cB$, some exchange must involve~$z_1$; in particular, it must exchange $z_1$ with~$z_2$.
Consider the first such exchange, and note that it does not change the value of $\beta$.
Before the exchange, EF1 requires $t\le \beta \le 2t$, and after the exchange, EF1 requires $0 \le \beta \le t$.
Hence, $\beta = t$, and the partition of~$G$ before (or after) the exchange corresponds to a solution of \textsc{Partition}.
It follows that the \textsc{Partition} instance is a Yes-instance.
\medskip

($\Rightarrow$)
Suppose that the \textsc{Partition} instance is a Yes-instance.
Let $L\subseteq\{1,\dots,r\}$ be such that $\sum_{\ell\in L}a_\ell = \sum_{\ell\not\in L}a_\ell$.

We construct an EF1 reconfiguration path as follows.
Starting from~$\cA$, for $\ell\in L$, let agent~$1$ exchange $g_\ell$ with an arbitrary good from~$H$ held by agent~$2$.
During this process, $\beta$ decreases from $2t$ to~$t$, and the allocation remains EF1 throughout.
Next, exchange $z_1$ with $z_2$.
Note that $\beta$ remains at~$t$, and EF1 is maintained.
Then, for $\ell\not\in L$, let agent~$1$ exchange $g_\ell$ with an arbitrary good from~$H$ held by agent~$2$.
During this process, $\beta$ decreases from $t$ to~$0$, and the allocation again remains EF1 throughout.
The final allocation reached is the target allocation~$\cB$.
It follows that the \textsc{EF1 Reconfiguration} instance is a Yes-instance.

\medskip

Finally, we extend our reduction to any fixed $n\ge 2$.
For each additional agent $i\in\{3,\dots,n\}$, we introduce a set $F_i$ of $r+5$ new goods, and give $F_i$ to agent~$i$ in both~$\cA$ and $\cB$.
The remaining goods are allocated between the two original agents in the same way as before.
Both original agents have utility~$0$ for all new goods.
Each new agent~$i$ has utility~$1$ for each good in $F_i$, utility~$0$ for each good in $F_j$ for $j\ne i$, utility $r+5$ for each of $x_1$ and $x_2$, and utility~$0$ for all other original goods.
Note that both the initial and target allocations are exactly-balanced, as each agent receives $r+5$ goods.
Clearly, agents~$1$ and $2$ do not envy any new agent, and the new agents do not envy one another.
Also, each new agent's utility for agent~$1$'s bundle after removing~$x_1$ is $r+5$, the same as her utility for her own bundle.
Hence, both allocations are EF1.

We show that an EF1 reconfiguration path exists in this extended instance if and only if such a path exists in the original two-agent instance.
Consider an EF1 reconfiguration path in the extended instance.
We claim that every exchange must be between agents $1$ and $2$ and involve two goods outside $X\cup Y$.
Assume otherwise, and consider the first exchange not of this type.
If it is between agents $1$ and $2$ (and involves at least one good in $X\cup Y$), the previous argument yields an EF1 violation.
Hence, the exchange must involve a new agent.
If agent~$1$ exchanges a good from~$X$ with a new agent, agent~$1$'s utility for her own bundle becomes at most $12T + 8t < 13T$, while her utility for agent~$2$'s bundle even after removing any good is at least $24T$.
This means that EF1 is violated for agent~$1$, a contradiction.
Similarly, if agent~$2$ exchanges a good from~$Y$ with a new agent, agent~$2$'s utility for her own bundle becomes at most $9T + 8t < 10T$, while her utility for agent~$1$'s bundle even after removing any good is at least $12T$.
This means that EF1 is violated for agent~$2$, a contradiction.
Lastly, suppose that a new agent~$i$ exchanges a good with another agent's good, where the latter good is outside $X\cup Y$.
After the exchange, agent~$i$'s utility for her own bundle is at most $r+4$, while her utility for agent~$1$'s bundle even after removing any good is at least $r+5$.
Hence, EF1 is violated for agent~$i$, again a contradiction.

Therefore, every exchange on an EF1 reconfiguration path in the extended instance must be between agents $1$ and $2$ and involve two goods outside $X\cup Y$.
This induces an EF1 reconfiguration path in the original two-agent instance.
Conversely, an EF1 reconfiguration path in the original instance also induces one in the extended instance, by keeping all goods in $X\cup Y$ as well as the bundles of agents $3,\dots,n$ fixed.
As in our arguments for $\cA$ and $\cB$, EF1 is maintained throughout the latter path.
This completes the proof.
\end{proof}

In light of \Cref{thm:NP-hard}, it is natural to ask whether \textsc{EF1 Reconfiguration} belongs to NP in the smallest case of two agents, as this would imply NP-completeness.
While we are unable to answer this question, we present a polynomial-time reduction from the problem \textsc{Subset Sum Reconfiguration} to \textsc{EF1 Reconfiguration} for $n = 2$.
In \textsc{Subset Sum Reconfiguration}, we are given a multiset of positive integers $\{a_1,\dots,a_r\}$, a lower threshold $\underline{t}$, an upper threshold $\overline{t}$, and---writing $R = \{1,\dots,r\}$---subsets $R_1,R_2\subseteq R$ such that both $\sum_{\ell\in R_1}a_\ell$ and $\sum_{\ell\in R_2}a_\ell$ belong to $[\underline{t}, \overline{t}]$.
The task is to determine whether $R_1$ can be transformed into~$R_2$ by adding or removing one index at a time in such a way that for each intermediate set~$R'$, it also holds that $\sum_{\ell\in R'}a_\ell$ belongs to $[\underline{t}, \overline{t}]$.
\citet{ItoDe14} showed that \textsc{Subset Sum Reconfiguration} is NP-hard, but left membership in NP open.
Our reduction in the following theorem implies that if \textsc{EF1 Reconfiguration} with two agents belongs to NP, then so does \textsc{Subset Sum Reconfiguration}---this would settle the question of \citet{ItoDe14}.

\begin{theorem}
\label{thm:polytime-reduction}
There exists a polynomial-time reduction from \textsc{Subset Sum Reconfiguration} to \textsc{EF1 Reconfiguration} with $n = 2$ agents and exactly-balanced initial and target allocations.
\end{theorem}

\begin{proof}
Consider any instance of \textsc{Subset Sum Reconfiguration}.
Let $b = \sum_{\ell\in R} a_\ell$ and $B = b + \overline{t} + 1$, let $c_\ell = 2^{\ell-1}$ for each $\ell\in R$, and let $C = \sum_{\ell\in R}c_\ell = 2^r - 1$.
For each $\ell\in R$, create two goods $g_\ell^0$ and $g_\ell^1$ valued identically by both agents: $u(g_\ell^0) = Bc_\ell$ and $u(g_\ell^1) = Bc_\ell + a_\ell$, where $u$ denotes the common utility function for these goods.
Let $\underline{T} = BC+\underline{t}$ and $\overline{T} = BC+\overline{t}$.
For any subset~$S$ of the $2r$ goods created thus far and any $\ell\in R$, let $q_\ell = |S\cap\{g_\ell^0, g_\ell^1\}|$, and let $f(S) = \sum_{\ell\in R}(q_\ell-1)c_\ell$.
Note that $(q_\ell-1)c_\ell$ is equal to $-c_\ell, 0, c_\ell$ when $q_\ell$ is equal to $0,1,2$, respectively.
We have 
\[
u(S) = B(C + f(S)) +\sum_{g_\ell^1\in S}a_\ell,
\]
which is an integer. 
If $f(S) \le -1$, then 
\[
u(S) \le BC-B + b = BC - \overline{t} - 1 < BC \le \underline{T}.
\]
On the other hand, if $f(S) \ge 1$, then
\[
u(S) \ge BC + B > \overline{T}.
\]
Hence, if $u(S) \in [\underline{T}, \overline{T}]$, then $f(S) = 0$.
Since the terms $c_\ell$ form a geometric sequence with ratio~$2$ and $q_\ell - 1\in\{-1,0,1\}$, this implies that $q_\ell = 1$ for all $\ell\in R$.
When $f(S) = 0$, we have $u(S) = BC + \sum_{g_\ell^1\in S} a_\ell$, and so $u(S) \in [\underline{T}, \overline{T}]$ if and only if $\sum_{g_\ell^1\in S} a_\ell \in [\underline{t}, \overline{t}]$.

Next, let $U = \sum_{\ell\in R}(u(g_\ell^0) + u(g_\ell^1)) = 2BC + b$, and let $\underline{\delta} = U - 2\underline{T}$ and $\overline{\delta} = 2\overline{T} - U$.
Let $Q = 4(|\underline{\delta}| + |\overline{\delta}| + U + 1)$.
Add eight goods in the sets $X = \{x_1,x_2\}$, $Y = \{y_1,y_2,y_3,y_4\}$, $Z = \{z_1,z_2\}$ with utilities as shown in the following table.
\begin{center} 
    \begin{tabular}{c|cccccccc}
        & $x_1$ & $x_2$ & $y_1$ & $y_2$ & $y_3$ & $y_4$ & $z_1$ & $z_2$ \\
        \hline
        $u_1$ & $6Q + \underline{\delta}$ & $6Q$ & $5Q$ & $4Q$ & $4Q$ & $4Q$ & $0$ & $0$  \\
        $u_2$ & $5Q$ & $4Q$ & $Q + \overline{\delta}$ & $Q$ & $Q$ & $Q$ & $0$ & $0$
    \end{tabular}
\end{center}
Note that all utilities are indeed nonnegative, and the total number of goods is $2r+8$.
We map any subset $R'\subseteq R$ with $\sum_{\ell\in R'}a_\ell \in [\underline{t}, \overline{t}]$ to the following allocation:
\begin{itemize}
\item agent~$1$ receives $x_1,x_2,z_1,z_2$, the goods $g_\ell^1$ for $\ell\in R'$, and the goods $g_\ell^0$ for $\ell\not\in R'$;
\item agent~$2$ receives $y_1,y_2,y_3,y_4$, the goods $g_\ell^1$ for $\ell\not\in R'$, and the goods $g_\ell^0$ for $\ell\in R'$.
\end{itemize}
In particular, each agent receives $r+4$ goods, so the allocation is exactly-balanced.
Consider any allocation such that agent~$1$ holds both goods from~$X$ but none of the goods from~$Y$.
Since $Q > U$, the utility for any of the first $2r$ goods is less than $Q$.
Therefore, agent~$1$ is EF1 if and only if $12Q + \underline{\delta} + u(S) \ge 12Q + (U - u(S))$.
This is equivalent to $2u(S) \ge U - \underline{\delta} = 2\underline{T}$, that is, $u(S) \ge \underline{T}$.
Similarly, agent~$2$ is EF1 if and only if $4Q + \overline{\delta} + (U-u(S)) \ge 4Q + u(S)$.
This is equivalent to $2u(S) \le U + \overline{\delta} = 2\overline{T}$, that is, $u(S) \le \overline{T}$.
Therefore, the initial and target allocations are EF1.
This completes the description of the \textsc{EF1 Reconfiguration} instance, which can be constructed in polynomial time; in particular, numbers up to $2^r$ can be written using a polynomial number of bits.

It remains to show that a Yes-instance of \textsc{Subset Sum Reconfiguration} corresponds to a Yes-instance of \textsc{EF1 Reconfiguration}, and vice versa.
If an allocation is such that agent~$1$ holds $x_1,x_2,z_1,z_2$ along with a set~$S$ containing exactly one of $g_\ell^0$ and $g_\ell^1$ for each $\ell\in R$, the argument above implies that the allocation is EF1 if and only if $\sum_{\ell\in R'}a_\ell\in [\underline{t}, \overline{t}]$, where $R' = \{\ell\in R \mid g_\ell^1\in S\}$.
If agent~$1$ exchanges a good from~$X$ with another good, then agent~$1$'s utility for her own bundle becomes at most $11Q + |\underline{\delta}| + U < 12Q$, while her utility for agent~$2$'s bundle even after removing any good is at least $12Q$, so EF1 is violated for agent~$1$.
Similarly, if agent~$2$ exchanges a good from~$Y$ with a good outside~$X$, then agent~$2$'s utility for her own bundle becomes at most $3Q+|\overline{\delta}|+U < 4Q$, while her utility for agent~$1$'s bundle even after removing any good is at least $4Q$, so EF1 is violated for agent~$2$.
This means that in order to maintain EF1, the goods from $X\cup Y$ cannot move.
Also, the condition $u(S) \in [\underline{T},\overline{T}]$ forces each agent to always hold exactly one of $g_\ell^0$ and $g_\ell^1$ for each $\ell\in R$.
Since the size vector of the allocation is fixed, $z_1,z_2$ must then stay with agent~$1$.
With the assignment of~$g_\ell^1$ to agent~$1$ corresponding to the inclusion of~$a_\ell$ (and the assignment of~$g_\ell^0$ to agent~$1$ corresponding to the exclusion of~$a_\ell$), this yields the desired correspondence between the two instances.
\end{proof}

\begin{corollary}
If \textsc{EF1 Reconfiguration} with $n = 2$ agents belongs to NP, then \textsc{Subset Sum Reconfiguration} also belongs to NP.
\end{corollary}

A possible approach to proving that \textsc{EF1 Reconfiguration} belongs to NP for two agents is to show that every Yes-instance admits an EF1 reconfiguration path of polynomial length, since such a path would serve as a certificate.
Although it remains unclear whether a polynomial-length path always exists, we show that a substantial detour may be unavoidable.
In particular, for two agents and $m$ goods, a shortest EF1 reconfiguration path can have length $\Theta(m^2)$, significantly exceeding the length of a shortest path when the EF1 constraint is ignored.

\begin{theorem}
\label{thm:path-length}
For any positive integer~$r$, there exists an instance with $n = 2$ agents and $4r+8$ goods, and exactly-balanced initial and target EF1 allocations, such that the minimum length of an EF1 reconfiguration path is $2r^2$.
\end{theorem}

\begin{proof}
Fix a positive integer~$r$, and let $R = \{0,1,\dots,4r-1\}$, $b = 4r^2 - r = \frac{1}{2}\sum_{\ell \in R} \ell$, and $c = b+2$.
For each $\ell\in R$, create a good~$g_\ell$ such that $u_1(g_\ell) = u_2(g_\ell) = c + \ell$.
Each agent's total utility for these $4r$ goods is $P \coloneqq 4rc + 2b$.
Let $Q = 10(P+2)$.
Add eight goods in the sets $X = \{x_1,x_2\}$, $Y = \{y_1,y_2,y_3,y_4\}$, $Z = \{z_1,z_2\}$ with utilities as shown in the following table.
\begin{center} 
    \begin{tabular}{c|cccccccc}
        & $x_1$ & $x_2$ & $y_1$ & $y_2$ & $y_3$ & $y_4$ & $z_1$ & $z_2$ \\
        \hline
        $u_1$ & $3Q$ & $3Q$ & $2Q$ & $2Q$ & $2Q$ & $2Q$ & $0$ & $0$  \\
        $u_2$ & $5Q$ & $4Q$ & $Q + 2$ & $Q$ & $Q$ & $Q$ & $0$ & $0$
    \end{tabular}
\end{center}
There are a total of $4r+8$ goods.
For a set $S\subseteq R$ of size~$2r$, let $A_1(S) = \{x_1,x_2,z_1,z_2\}\cup\{g_\ell\mid \ell\in S\}$ and $A_2(S) = \{y_1,y_2,y_3,y_4\} \cup \{g_\ell\mid \ell\not\in S\}$, so the allocation $\cA(S) = (A_1(S), A_2(S))$ is exactly-balanced for any~$S$.
Let $S_0 = \{0,\dots,r-1\}\cup\{3r,\dots,4r-1\}$, and let the initial and target allocations be $\cA(S_0)$ and $\cA(R\setminus S_0)$, respectively.

Consider any allocation $\cA = (A_1,A_2)$ such that $x_1,x_2$ are with agent~$1$ while $y_1,y_2,y_3,y_4$ are with agent~$2$, and let $S$ be the subset of indices $\ell\in\{0,\dots,4r-1\}$ (not necessarily of size~$2r$) such that~$g_\ell$ is held by agent~$1$.
Let $u(S) = c|S| + \sum_{\ell\in S}\ell$ be the common utility that both agents have for the corresponding goods.
Agent~$1$'s utility for her own bundle is $6Q + u(S)$, while her utility for agent~$2$'s bundle after removing a most valuable good is $6Q + P - u(S)$.
Hence, agent~$1$ is EF1 exactly when $2u(S) \ge P$.
Similarly, agent~$2$'s utility for her own bundle is $4Q+2+P-u(S)$, while her utility for agent~$1$'s bundle after removing the most valuable good is $4Q + u(S)$.
Hence, agent~$2$ is EF1 exactly when $2u(S) \le P+2$.
This means that $\cA$ is EF1 if and only if $u(S) \in [P/2, P/2 + 1]$.
Now, $P/2 = 2rc + b$.
If $|S| \le 2r-1$, then
\[
u(S) \le (2r-1)c + 2b = 2rc + b - 2 < \frac{P}{2}.
\]
On the other hand, if $|S| \ge 2r+1$, then
\[
u(S) \ge (2r+1)c = 2rc + b + 2 > \frac{P}{2} + 1.
\]
Therefore, if $u(S) \in [P/2, P/2 + 1]$, then $|S| = 2r$.
This implies that for any EF1 allocation~$\cA$ with the assignment of $x_1,x_2,y_1,y_2,y_3,y_4$ fixed as above, we must have $|S| = 2r$.
If this allocation is on a reconfiguration path between $\cA(S_0)$ and $\cA(R\setminus S_0)$, the size vector forces $z_1,z_2$ to be with agent~$1$.
When $|S| = 2r$, the condition $u(S) \in [P/2, P/2+1]$ becomes $\sum_{\ell\in S}\ell \in \{b, b+1\}$.
Since $\sum_{\ell\in S_0}\ell = \sum_{\ell\in R\setminus S_0}\ell = b$, both the initial allocation $\cA(S_0)$ and the target allocation $\cA(R\setminus S_0)$ are EF1.

We now consider any EF1 reconfiguration path between the initial and target allocations.
We claim that none of the goods in $X\cup Y$ can move.
Assume otherwise, and consider the first exchange involving at least one of these goods.
If agent~$1$ exchanges a good from~$X$ with another good, then agent~$1$'s utility for her own bundle becomes at most $5Q + P < 6Q$, while her utility for agent~$2$'s bundle even after removing any good is at least $6Q$, so EF1 is violated for agent~$1$.
Similarly, if agent~$2$ exchanges a good from~$Y$ with a good outside~$X$, then agent~$2$'s utility for her own bundle becomes at most $3Q+2+P < 4Q$, while her utility for agent~$1$'s bundle even after removing any good is at least~$4Q$, so EF1 is violated for agent~$2$.
Hence, none of the goods in $X\cup Y$ can move.
By our argument in the previous paragraph, in order for EF1 to be maintained, it must always hold that $|S| = 2r$ and $\sum_{\ell\in S}\ell \in \{b, b+1\}$.
Also, since the size vector is fixed, $z_1$ and $z_2$ must stay with agent~$1$.

Initially, $S_0 = \{0,\dots,r-1\}\cup\{3r,\dots,4r-1\}$.
In order to maintain $\sum_{\ell\in S}\ell \in \{b, b+1\}$, in each exchange, we can only replace some~$\ell$ by either $\ell - 1$ or $\ell+1$.
Hence, to arrive at $R\setminus S_0 = \{r,\dots,3r-1\}$, the number of exchanges required is at least
\begin{align*}
|r-0| + \dots + |(2r-1)-(r-1)| + |2r-3r| + \dots + |(3r-1) - (4r-1)| 
= 2r^2.
\end{align*}
It remains to show that $2r^2$ exchanges suffice.
For each $(t, t')\in\{0,\dots,r-1\}\times \{0,\dots,r-1\}$ in increasing order of~$t$ and then of~$t'$, we replace $t+r-t'-1$ with $t+r-t'$ and then replace $3r-t+t'$ with $3r-t+t'-1$.
For each $(t, t')$, this increases $\sum_{\ell\in S}\ell$ from $b$ to $b+1$ and then decreases it back to~$b$.
Therefore, the allocation remains EF1 throughout, and after $2r^2$ exchanges, the target allocation~$\cA(R\setminus S_0)$ is reached.
\end{proof}

\section{Exchanges and Transfers}
\label{sec:transfers}

In this section, we turn to a setting in which the allowed operations include not only an exchange, but also a \emph{transfer} of a good from one agent to another.
We refer to this as the \emph{exchange-and-transfer setting}.
With transfers, allocations no longer need to have the same size vector in order to be reachable from one another; we say that two allocations are \emph{adjacent} if they can be reached from each other via either an exchange or a transfer.
\citet[App.~A]{IgarashiKaSu24} initiated the study of this model and showed that allowing transfers leads to stronger positive results than in the exchange-only setting.
We shall delve further into this setting by exploring both its possibilities and limitations.

To begin with, \citet{IgarashiKaSu24} proved that when there are either two agents with identical utilities or any number of agents with identical binary utilities, an EF1 reconfiguration path exists between any two EF1 allocations in this setting.
We strengthen both results simultaneously by proving the same guarantee for any number of agents with identical utilities.
This stands in stark contrast to the exchange-only setting, where these authors showed that an EF1 reconfiguration path may fail to exist even for three agents.

\begin{theorem}
\label{thm:transfer-identical}
In the exchange-and-transfer setting, for an instance with any number of agents having identical utilities and any initial and target EF1 allocations, there exists an EF1 reconfiguration path.
\end{theorem}

\begin{proof}
Let $u$ denote the common utility function of the agents.
For any bundle~$B$, let $\rho(B)$ denote the utility of~$B$ after removing a highest-value good (if $B = \emptyset$, let $\rho(B) = 0$).
For any allocation~$\cA$, let $L(\cA) = \min_{i\in N} u(A_i)$ and $R(\cA) = \max_{i\in N} \rho(A_i)$.
Then, an allocation~$\cA$ satisfies EF1 if and only if $L(\cA) \ge R(\cA)$.
Fix a good~$g$ of maximum utility.
We say that an allocation~$\cA$ is \emph{$g$-ready} if the owner~$i$ of~$g$ satisfies $\rho(A_i) = u(A_i\setminus\{g\}) = R(\cA)$.
Observe that if an EF1 allocation $\cA$ is $g$-ready, then removing~$g$ still leaves an EF1 allocation.
Indeed, every agent still has utility at least $R(\cA)$, and each agent's utility after removing her highest-value good is at most $R(\cA)$.
Moreover, after removing~$g$, agent~$i$ has the lowest utility among all agents.

We claim that from any EF1 allocation $\cA$, it is possible to reach a $g$-ready allocation via an EF1 reconfiguration path.
Let agent~$i$ be the owner of~$g$ in~$\cA$.
If $u(A_i\setminus\{g\}) = R(\cA)$, the allocation~$\cA$ is already $g$-ready.
Else, $u(A_i\setminus\{g\}) < R(\cA)$; let $j$ be an agent such that $\rho(A_j) = R(\cA)$, and let $h$ be a highest-value good in~$A_j$, so that $u(A_j\setminus\{h\}) = R(\cA)$.
We consider two cases.
\begin{itemize}
\item \underline{Case 1}: $u(A_i\setminus\{g\}) + u(h) \ge R(\cA)$.
We exchange $g$ and $h$, and call the new allocation~$\cB$.
It holds that $\rho(B_j) = u(B_j\setminus\{g\}) = R(\cA)$, $u(B_i) = u(A_i\setminus\{g\}) + u(h) \ge R(\cA)$, and $\rho(B_i) \le u(B_i\setminus\{h\}) = u(A_i\setminus\{g\}) < R(\cA)$.
Hence, $L(\cB) \ge R(\cA) \ge R(\cB)$, and so $\cB$ satisfies EF1.
Moreover, $\cB$ is $g$-ready.

\item \underline{Case 2}: $u(A_i\setminus\{g\}) + u(h) < R(\cA)$.
We transfer $h$ from~$j$ to~$i$, and call the new allocation~$\cB$.
It holds that $u(B_j) = R(\cA)$ and $\rho(B_i) = u(B_i\setminus\{g\}) = u(A_i\setminus\{g\}) + u(h) < R(\cA)$.
Again, $L(\cB) \ge R(\cA) \ge R(\cB)$, and so $\cB$ satisfies EF1.
We then repeat the process for~$\cB$ instead of~$\cA$, noting that agent~$i$ remains the owner of~$g$.
Since Case~2 involves transferring a good to~$i$, it can only happen a finite number of times, and the process must eventually stop with a $g$-ready allocation.
\end{itemize}

Observe that given an EF1 allocation, assigning an extra good to an agent with the minimum utility preserves EF1.\footnote{This can be viewed as a step in the envy cycle elimination algorithm for identical utilities \citep{LiptonMaMo04}.}
For an EF1 allocation~$\cA$ of $M\setminus\{g\}$ and an agent~$i$ with the minimum utility in~$\cA$, we write $\cA^{(i)}$ for the allocation resulting from assigning $g$ to~$i$.
Let $\cA,\cB$ be adjacent\footnote{Recall the definition from the beginning of this section.} EF1 allocations of $M\setminus\{g\}$.
We claim that any $\cA^{(i)}$ and $\cB^{(j)}$---which are necessarily EF1---can be connected via an EF1 reconfiguration path of length at most~$3$.
If some agent~$i^*$ is a minimum-utility agent in both $\cA$ and $\cB$, we transfer $g$ from~$i$ to~$i^*$ (if $i\ne i^*$), perform the operation connecting $\cA$ to~$\cB$, and transfer~$g$ from~$i^*$ to~$j$ (if $i^* \ne j$); both $\cA^{(i^*)}$ and $\cB^{(i^*)}$ are EF1 by the earlier observation.
Else, no agent is a minimum-utility agent in both $\cA$ and $\cB$, so $i\ne j$.
Since $i$ is a minimum-utility agent in~$\cA$, we have $L(\cA) = u(A_i)$; similarly, $L(\cB) = u(B_j)$.
Assume that $L(\cA) \ge L(\cB)$; the opposite case can be handled symmetrically.
The unique minimum-utility agent in~$\cB$ must be the agent whose utility decreases from~$\cA$ to~$\cB$.

Consider the operation that transforms $\cA$ into~$\cB$.
Suppose first that it is a transfer of a good~$h$ from agent~$i'$ to agent~$j'$.
Recall that $\cA$ and $\cB$ are allocations of $M\setminus\{g\}$, so $h\ne g$.
Since $i'$ is the unique minimum-utility agent in~$\cB$, we have $i' = j$.
If $i = j'$, then from~$\cA^{(i)}$, we exchange $g$ and~$h$ and immediately arrive at $\cB^{(j)}$.
Suppose now that $i\not\in\{j,j'\}$.
From~$\cA^{(i)}$, we exchange $g$ and $h$ and then transfer $h$ to agent~$j'$, arriving at $\cB^{(j)}$ (see \Cref{fig:transfer-case-theorem-5-1}).
We argue that the intermediate allocation~$\cC$ is EF1.
Since every agent receives at least as much utility in~$\cC$ as in $\cA$ (i.e., $\cA^{(i)}$ without~$g$) and agent~$i$ is a minimum-utility agent in~$\cA$, we have $L(\cC) \ge L(\cA) = u(A_i)$.
Since $\cA$ is EF1 and $C_\ell = A_\ell$ for all $\ell\not\in\{i,j\}$, it holds that $\rho(C_\ell) = \rho(A_{\ell}) \le L(\cA)$.
Also, $\rho(C_i) = \rho(A_i\cup\{h\}) \le u(A_i) = L(\cA)$, and since $g$ is a highest-value good and agent~$j$ is a minimum-utility agent in~$\cB$, we have $\rho(C_j) = \rho(B_j\cup\{g\}) = u(B_j) \le u(B_i) = u(A_i)$.
It follows that $L(\cC) \ge u(A_i) \ge R(\cC)$, and $\cC$ is EF1.

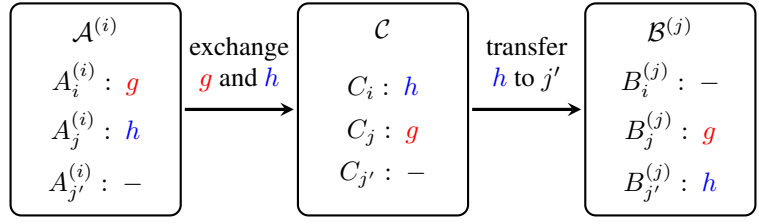
\begin{figure}[t]
\centering
\begin{tikzpicture}[
    font=\small,
    >=stealth,
    panel/.style={
        draw,
        rounded corners,
        thick,
        minimum width=2.25cm,
        minimum height=2.8cm,
        inner sep=2pt
    }
]

\node[panel] (P1) at (0,0) {};
\node at ($(P1.north)+(0,-0.35)$) {$\cA^{(i)}$};

\node[align=center, inner sep=0pt] at ($(P1.center)+(0,-0.30)$) {%
    $A^{(i)}_i:\ \textcolor{red}{g}$\\[4pt]
    $A^{(i)}_j:\ \textcolor{blue}{h}$\\[4pt]
    $A^{(i)}_{j'}:\ -$
};

\node[panel] (P2) at (3.8,0) {};
\node at ($(P2.north)+(0,-0.35)$) {$\cC$};

\node[align=center, inner sep=0pt] at ($(P2.center)+(0,-0.30)$) {%
    $C_i:\ \textcolor{blue}{h}$\\[4pt]
    $C_j:\ \textcolor{red}{g}$\\[4pt]
    $C_{j'}:\ -$
};

\node[panel] (P3) at (7.60,0) {};
\node at ($(P3.north)+(0,-0.35)$) {$\cB^{(j)}$};

\node[align=center, inner sep=0pt] at ($(P3.center)+(0,-0.30)$) {%
    $B^{(j)}_i:\ -$\\[4pt]
    $B^{(j)}_j:\ \textcolor{red}{g}$\\[4pt]
    $B^{(j)}_{j'}:\ \textcolor{blue}{h}$
};

\draw[->,very thick]
    ($(P1.east)+(0.04,0)$) -- ($(P2.west)+(-0.04,0)$);
\node[align=center] at (1.9,0.6)
    {exchange\\ $\textcolor{red}{g}$ and $\textcolor{blue}{h}$};

\draw[->,very thick]
    ($(P2.east)+(0.04,0)$) -- ($(P3.west)+(-0.04,0)$);
\node[align=center] at (5.7,0.6)
    {transfer\\ $\textcolor{blue}{h}$ to $j'$};

\end{tikzpicture}
\caption{Illustration of the transfer case in the proof of \Cref{thm:transfer-identical} (subcase $i\notin\{j,j'\}$; recall that $j = i'$). 
Each good different from $g$ and $h$ belongs to the same agent in all three allocations and is not displayed.
Starting from~$\cA^{(i)}$, we exchange $g$ and $h$ to obtain $\cC$, and then transfer $h$ to agent~$j'$ to obtain $\cB^{(j)}$.}
\label{fig:transfer-case-theorem-5-1}
\end{figure}

Next, suppose that the operation that transforms $\cA$ into~$\cB$ is an exchange between a good~$h_1$ of agent~$i'$ and a good~$h_2$ of agent~$j'$.
If $u(h_1) = u(h_2)$, then $\cA$ and $\cB$ share a minimum-utility agent, a contradiction, so $u(h_1) \ne u(h_2)$.
Assume without loss of generality that $u(h_1) < u(h_2)$.
Thus, the utility of agent~$j'$ decreases, which means that $j'$ is the unique minimum-utility agent in~$\cB$, so $j' = j$.
\begin{itemize}
\item \underline{Case 1}: $i = i'$.
From~$\cA^{(i)}$, we exchange $g$ and $h_2$ and then transfer $h_1$ to agent~$j$, arriving at~$\cB^{(j)}$ (see \Cref{fig:exchange-case1-theorem-5-1}).
Consider the intermediate allocation~$\cC$.
We have $u(C_i) = u(A_i\cup\{h_2\}) \ge u(A_i) = L(\cA)$ and $u(C_j) = u(A_j\setminus\{h_2\}) + u(g) \ge u(A_j) \ge L(\cA)$.
Also, $\rho(C_i) \le u(C_i\setminus\{h_2\}) = u(A_i) = L(\cA)$ and $\rho(C_j) = u(A_j\setminus\{h_2\}) \le u(B_j) = L(\cB) \le L(\cA)$.
Hence, $\cC$ is EF1.

\begin{figure}[t]
\centering
\begin{tikzpicture}[
    font=\small,
    >=stealth,
    panel/.style={
        draw,
        rounded corners,
        thick,
        minimum width=2.3cm,
        minimum height=2.2cm,
        inner sep=2pt
    }
]

\node[panel] (P1) at (0,0) {};
\node at ($(P1.north)+(0,-0.35)$) {$\cA^{(i)}$};

\node[align=center, inner sep=0pt] at ($(P1.center)+(0,-0.25)$) {%
    $A^{(i)}_i:\ \textcolor{red}{g},\ \textcolor{blue}{h_1}$\\[4pt]
    $A^{(i)}_j:\ \textcolor{violet}{h_2}$
};

\node[panel] (P2) at (4.2,0) {};
\node at ($(P2.north)+(0,-0.35)$) {$\cC$};

\node[align=center, inner sep=0pt] at ($(P2.center)+(0,-0.25)$) {%
    $C_i:\ \textcolor{blue}{h_1},\ \textcolor{violet}{h_2}$\\[4pt]
    $C_j:\ \textcolor{red}{g}$
};

\node[panel] (P3) at (8.4,0) {};
\node at ($(P3.north)+(0,-0.35)$) {$\cB^{(j)}$};

\node[align=center, inner sep=0pt] at ($(P3.center)+(0,-0.25)$) {%
    $B^{(j)}_i:\ \textcolor{violet}{h_2}$\\[4pt]
    $B^{(j)}_j:\ \textcolor{red}{g},\ \textcolor{blue}{h_1}$
};

\draw[->,very thick]
    ($(P1.east)+(0.04,0)$) -- ($(P2.west)+(-0.04,0)$);
\node[align=center] at (2.1,0.6)
    {exchange\\ $\textcolor{red}{g}$ and $\textcolor{violet}{h_2}$};

\draw[->,very thick]
    ($(P2.east)+(0.04,0)$) -- ($(P3.west)+(-0.04,0)$);
\node[align=center] at (6.3,0.6)
    {transfer\\ $\textcolor{blue}{h_1}$ to $j$};

\end{tikzpicture}
\caption{Illustration of Case~1 (i.e., $i = i'$; recall that $j = j'$) of the exchange case in the proof of \Cref{thm:transfer-identical}. 
Each good different from $g$, $h_1$, and $h_2$ belongs to the same agent in all three allocations and is not displayed.
Starting from $\cA^{(i)}$, we exchange $g$ and $h_2$ to obtain $\cC$, and then transfer $h_1$ to agent~$j$ to obtain $\cB^{(j)}$.}
\label{fig:exchange-case1-theorem-5-1}
\end{figure}
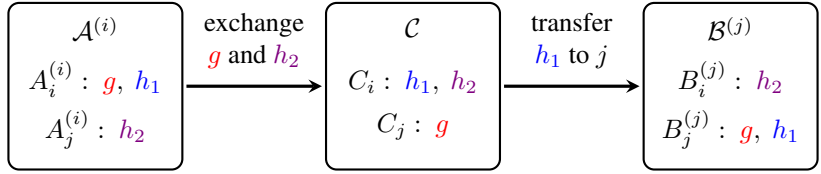

\item \underline{Case 2}: $i\not\in\{i',j\}$. 
In this case, from~$\cA^{(i)}$, we exchange $g$ and $h_2$, then exchange $h_2$ and $h_1$, and finally transfer $h_1$ to agent~$j$, arriving at $\cB^{(j)}$ (see \Cref{fig:exchange-case2-theorem-5-1}).
Consider the respective intermediate allocations $\cC$ and $\cD$.

For $\cC$, we have $u(C_i) \ge u(A_i) = L(\cA)$ and $u(C_j) \ge u(A_j) \ge u(A_i) = L(\cA)$.
Moreover, $\rho(C_i) = \rho(A_i\cup\{h_2\}) \le u(A_i) = L(\cA)$ and $\rho(C_j) = u(A_j\setminus\{h_2\}) = u(B_j\setminus\{h_1\}) \le u(B_j) = L(\cB) \le L(\cA)$.
Hence, $\cC$ is EF1.

For~$\cD$, the bundle of agent~$j$ is unchanged from~$\cC$, and we have $u(D_i) \ge u(A_i) = L(\cA)$ and $u(D_{i'}) > u(A_{i'}) \ge u(A_i)$.
Moreover, $\rho(D_i) = \rho(A_i\cup\{h_1\}) \le u(A_i) = L(\cA)$ and $\rho(D_{i'}) = \rho(B_{i'}) \le L(\cB) \le L(\cA)$.
Hence, $\cD$ is EF1 as well.
\end{itemize}
It follows that $\cA^{(i)}$ and $\cB^{(j)}$ can be connected via an EF1 reconfiguration path of length at most~$3$.

\begin{figure}[t]
\centering
\begin{tikzpicture}[
    font=\small,
    >=stealth,
    panel/.style={
        draw,
        rounded corners,
        thick,
        minimum width=2.3cm,
        minimum height=2.8cm,
        inner sep=2pt
    }
]

\node[panel] (P1) at (0,0) {};
\node at ($(P1.north)+(0,-0.35)$) {$\cA^{(i)}$};

\node[align=center, inner sep=0pt] at ($(P1.center)+(0,-0.30)$) {%
    $A^{(i)}_i:\ \textcolor{red}{g}$\\[4pt]
    $A^{(i)}_{i'}:\ \textcolor{blue}{h_1}$\\[4pt]
    $A^{(i)}_j:\ \textcolor{violet}{h_2}$
};

\node[panel] (P2) at (4,0) {};
\node at ($(P2.north)+(0,-0.35)$) {$\cC$};

\node[align=center, inner sep=0pt] at ($(P2.center)+(0,-0.30)$) {%
    $C_i:\ \textcolor{violet}{h_2}$\\[4pt]
    $C_{i'}:\ \textcolor{blue}{h_1}$\\[4pt]
    $C_j:\ \textcolor{red}{g}$
};

\node[panel] (P3) at (8,0) {};
\node at ($(P3.north)+(0,-0.35)$) {$\cD$};

\node[align=center, inner sep=0pt] at ($(P3.center)+(0,-0.30)$) {%
    $D_i:\ \textcolor{blue}{h_1}$\\[4pt]
    $D_{i'}:\ \textcolor{violet}{h_2}$\\[4pt]
    $D_j:\ \textcolor{red}{g}$
};

\node[panel] (P4) at (12,0) {};
\node at ($(P4.north)+(0,-0.35)$) {$\cB^{(j)}$};

\node[align=center, inner sep=0pt] at ($(P4.center)+(0,-0.30)$) {%
    $B^{(j)}_i:\ -$\\[4pt]
    $B^{(j)}_{i'}:\ \textcolor{violet}{h_2}$\\[4pt]
    $B^{(j)}_j:\ \textcolor{red}{g},\ \textcolor{blue}{h_1}$
};

\draw[->,very thick]
    ($(P1.east)+(0.04,0)$) -- ($(P2.west)+(-0.04,0)$);
\node[align=center] at (2,0.6)
    {exchange\\ $\textcolor{red}{g}$ and $\textcolor{violet}{h_2}$};

\draw[->,very thick]
    ($(P2.east)+(0.04,0)$) -- ($(P3.west)+(-0.04,0)$);
\node[align=center] at (6,0.6)
    {exchange\\ $\textcolor{violet}{h_2}$ and $\textcolor{blue}{h_1}$};

\draw[->,very thick]
    ($(P3.east)+(0.04,0)$) -- ($(P4.west)+(-0.04,0)$);
\node[align=center] at (10,0.6)
    {transfer\\ $\textcolor{blue}{h_1}$ to $j$};

\end{tikzpicture}
\caption{Illustration of Case~2 (i.e., $i\notin\{i',j\}$; recall that $j = j'$) of the exchange case in the proof of \Cref{thm:transfer-identical}. 
Each good different from $g$, $h_1$, and $h_2$ belongs to the same agent in all four allocations and is not displayed.
Starting from $\cA^{(i)}$, we exchange $g$ and $h_2$ to obtain $\cC$, then exchange $h_2$ and $h_1$ to obtain $\cD$, and finally transfer $h_1$ to agent~$j$ to obtain $\cB^{(j)}$.}
\label{fig:exchange-case2-theorem-5-1}
\end{figure}

We are now ready to prove the theorem.
The proof proceeds by induction on the number of goods.
If there is one good, the statement holds trivially.
Assume that the statement is true for $m-1$ goods, and let there be $m$~goods.
Consider any initial allocation~$\cA$ and target allocation~$\cB$.
First, we turn $\cA$ and~$\cB$ into $g$-ready allocations $\cA^*$ and $\cB^*$, respectively, via EF1 reconfiguration paths.
Removing $g$ leaves allocations $\cA^-$ and $\cB^-$, both of which are also EF1 since $\cA^*$ and $\cB^*$ are $g$-ready.
If $\cA^- = \cB^-$, then $\cA^*$ and $\cB^*$ are either identical or can reach each other via a single transfer.
Assume that $\cA^- \ne \cB^-$.
By the induction hypothesis, there exists an EF1 reconfiguration path connecting $\cA^-$ and~$\cB^-$.
Add $g$ to each allocation on this path.
In particular, for $\cA^-$ and~$\cB^-$, add $g$ to its owners in $\cA^*$ and $\cB^*$, respectively; the fact that $\cA^*$ and $\cB^*$ are EF1 guarantees that these are respective minimum-utility agents in $\cA^-$ and~$\cB^-$.
For the intermediate allocations, add $g$ to any minimum-utility agent.
Now, any two consecutive allocations on this path can be connected via an EF1 reconfiguration path of length at most~$3$.
Therefore, we can connect $\cA^*$ and $\cB^*$ via a reconfiguration path, completing the induction.
\end{proof}

By contrast, for binary utilities, we obtain a negative result: an EF1 reconfiguration path between two EF1 allocations does not always exist.

\begin{theorem}
\label{thm:transfer-binary}
In the exchange-and-transfer setting, for any $n\ge 5$, there exists an instance with $n$~agents having binary utilities together with initial and target EF1 allocations with the same size vector such that no EF1 reconfiguration path exists.
\end{theorem}

\begin{proof}
For convenience, we denote the set of agents by $N = \{0,1,\dots,n-1\}$, where $n\ge 5$.
Let $C = \{0,1,2,3,4\}$ be the set of \emph{core agents} and $D = N\setminus C$ be the set of \emph{extra agents}.

For each unordered pair $\{i,j\}\subseteq C$, create an \emph{edge good} $g_{i,j} = g_{j,i}$; there are $10$ edge goods in total.
For each $i\in C$, let $M_i = \{g_{i-1,i+1}, g_{i-2,i+2}\}$, where indices are taken modulo~$5$.
That is, 
\begin{align*}
&M_0 = \{g_{1,4}, g_{2,3}\}, \quad
M_1 = \{g_{0,2}, g_{3,4}\}, \quad
M_2 = \{g_{1,3}, g_{0,4}\}, \quad \\
&M_3 = \{g_{2,4}, g_{0,1}\}, \quad
M_4 = \{g_{0,3}, g_{1,2}\}.
\end{align*}
For each extra agent $i\in D$, create an \emph{anchor good} $h_i$.
Hence, the total number of goods is $10+(n-5) = n+5$.
The binary utilities are defined as follows.
\begin{itemize}
\item Each core agent $i\in C$ values an edge good $g_{j,j'}$ if and only if $i\in \{j,j'\}$.
\item Each core agent values every anchor good.
\item Each extra agent $i\in D$ values every edge good and only her anchor good~$h_i$.
\end{itemize}

The initial allocation~$\cA$ is given by $A_i = M_i$ for $i\in C$, and $A_i = \{h_i\}$ for $i\in D$.
Note that each core agent has utility at most~$1$ for any bundle, while each extra agent has utility~$1$ for her own bundle and at most~$2$ for any other bundle.
Hence, $\cA$ is EF1.
The target allocation~$\cB$ is given by $B_i = M_{i+1}$ for $i\in C$, where indices are taken modulo~$5$, and $B_i = \{h_i\}$ for $i\in D$.
By the same reasoning, $\cB$ is EF1 as well.

We claim that there does not exist an EF1 reconfiguration path between $\cA$ and~$\cB$.
In fact, we prove the stronger statement that no exchange or transfer from~$\cA$ produces another EF1 allocation.

First, consider a transfer from~$\cA$.
Suppose that a core agent transfers an edge good~$g_{i,i'}$ to another core agent~$j'$. 
Let $j\in \{i,i'\}\setminus\{j'\}$, so agent~$j$'s utility for her own bundle remains~$0$.
Since $j\ne j'$, before the transfer, agent~$j$ already has utility~$1$ for the bundle of agent~$j'$; after the transfer, this utility increases to~$2$.
Hence, the resulting allocation violates EF1.
If an edge good~$g_{i,i'}$ is transferred to an extra agent~$j$, then agent~$i$'s utility for her own bundle is~$0$ while her utility for agent~$j$'s bundle becomes~$2$, so EF1 is violated.
Lastly, if an anchor good is transferred, then its owner's utility for her own bundle drops to~$0$, while her utility for any core bundle remains at least~$2$, again violating EF1.

Next, consider an exchange from~$\cA$.
To begin with, suppose that two core agents $i$ and~$i'$ exchange edge goods.
We claim that after the exchange, one of these agents will hold two edge goods with overlapping indices.
Assume for contradiction that this is not the case.
Since agent~$i$ originally holds two goods covering all four indices other than~$i$, she must receive a good with index~$i$ from the exchange.
Also, the indices of the good that remains with agent~$i'$ form a subset of the union between the indices of the good that remains with agent~$i$ and the singleton set of index~$i$---since these two goods differ, the former good must contain index~$i$.
However, this implies that agent~$i'$ originally holds two goods with index~$i$ before the exchange, a contradiction.
Hence, there exists a core agent~$j'$ who holds two edge goods with overlapping indices, say index~$j$, after the exchange with a core agent~$j''$.
Since agent~$j$ originally does not hold any good with index~$j$, we have $j\not\in\{j',j''\}$.
Therefore, agent~$j$ still has utility~$0$ and envies agent~$j'$ by more than one good after the exchange.

Finally, consider an exchange involving an anchor good.
If the exchange is between two anchor goods, then the utility of an involved agent drops to~$0$ while her utility for a core agent's bundle remains~$2$, violating EF1.
Else, the exchange is between an edge good~$g$ and an anchor good~$h$.
After the exchange, the corresponding core agent holds $h$ along with her other edge good---let~$i$ be an index of this edge good.
Agent~$i$ has utility~$2$ for this agent's bundle but utility~$0$ for her own bundle, so EF1 is again violated.

In conclusion, no exchange or transfer from~$\cA$ produces another EF1 allocation, as claimed.
\end{proof}

Despite this negative result, we show next that if we relax the requirement on intermediate allocations to EF2, then a reconfiguration path always exists for binary utilities.
In fact, this holds even if we disallow exchanges---we call this setting the \emph{transfer-only setting}.

\begin{theorem}
\label{thm:transfer-EF2}
In the transfer-only setting, for an instance with any number of agents having binary utilities and any initial and target EF1 allocations, there exists an EF2 reconfiguration path.
\end{theorem}

\begin{proof}
If there are goods that yield utility~$0$ to every agent, we can transfer them to their target owners without changing any utilities.
Hence, we may ignore such goods and assume from now on that every good is valued by at least one agent.
Call an allocation~$\cA$ \emph{clean} if $u_i(A_i) = |A_i|$ for all $i\in N$, that is, every good is assigned to an agent who values it.\footnote{This property is often considered in work on binary or matroid-rank utilities, and is sometimes called \emph{non-redundant} or \emph{non-wasteful} \citep{BabaioffEzFe21,BenabbouChIg21,BarmanVe22,SuksompongTe23,MontanariScSu25}.} 
The proof first connects each of the initial and target allocations to a clean EF1 allocation, then connects the resulting clean allocations through allocations minimizing a common potential.

First, we show that any EF1 allocation can be transformed into a clean EF1 allocation via transfers in such a way that every intermediate allocation is EF2.
Let $\cA$ be the current EF1 allocation, and let $d_i = u_i(A_i)$ for each $i\in N$.
A good is called \emph{misallocated} if it yields utility~$0$ to the current owner; let $D$ be the set of misallocated goods.
Suppose that $D\neq \emptyset$.
We create a bipartite graph between $D$ and the set of agents~$N$ where there is an edge between a good and each agent who values it.
Note that each good has an edge to at least one agent but not to its current owner.
Choose a matching of maximum size, and among such matchings, choose one that minimizes the sum of $d_i$ across all matched agents~$i$.
Transfer each matched good to its matched agent in arbitrary order, and let $\cB$ be the resulting allocation after these transfers.

We claim that every intermediate allocation is EF2, and $\cB$ is EF1.
During the transfer process, no agent loses utility since all transferred goods were originally misallocated.
Also, each agent receives at most one additional good.
Since the original allocation~$\cA$ is EF1, every intermediate allocation is EF2.
We now show that $\cB$ is EF1.
For each agent who receives a good, her utility for her own bundle increases by~$1$ and her utility for any other agent's bundle increases by at most~$1$, so she remains EF1.
Consider an unmatched agent~$i$; her utility remains~$d_i$.
The only possible violation concerns a matched agent~$j$ whose incoming good~$g$ is valued by~$i$.
In this case, every good in~$D$ valued by~$i$ must already be matched, since otherwise we can match such a good with~$i$ and increase the size of the matching.
Hence, none of the goods from~$D$ remaining in $j$'s bundle (if any) is valued by~$i$.
Note that there are exactly $d_j$~goods originally held by~$j$ and valued by~$j$.
Also, we must have $d_j\le d_i$, since otherwise we can match~$g$ with~$i$ instead and decrease the sum objective of the matching.
Therefore, after the transfer process, $i$'s utility for $j$'s bundle is at most $d_j + 1 \le d_i + 1 = u_i(B_i) + 1$, and so EF1 is satisfied for agent~$i$.
By repeating this matching and transfer process as long as there are misallocated goods, we arrive at a clean EF1 allocation.

Next, for a clean allocation~$\cA$, we define its \emph{transfer graph} as a directed graph with the set of agents~$N$ as the set of vertices, and there is an edge from an agent~$i$ to another agent~$j$ if $u_i(A_j) > 0$.
Consider any simple directed path $i_0\rightarrow i_1\rightarrow\dots\rightarrow i_t$ for some positive integer~$t$.
For each $r\in\{1,\dots,t\}$, let $g_r$ be a good in~$A_{i_r}$ such that $u_{i_{r-1}}(g_r) = 1$.
We transfer the goods in backward order along the path, starting with $g_t$ from $i_t$ to $i_{t-1}$.
Since each agent values a good that she receives, every agent has the same utility as before, except that agent~$i_t$'s utility decreases by~$1$ and agent~$i_0$'s utility increases by~$1$.
We claim that if $\cA$ is EF1 and agent~$i_t$ has the maximum bundle size on the path, then every intermediate allocation is EF2.
Indeed, other than agent~$i_t$, no agent's utility decreases during the process (relative to her original utility), so since each bundle receives at most one additional good, every such agent remains EF2.
For agent~$i_t$, her utility decreases from~$d_{i_t}$ to $d_{i_t} - 1$.
Since $\cA$ is EF1, her utility for any bundle outside the path remains at most $d_{i_t} + 1$, and a bundle on the path has size at most $d_{i_t} + 1$ due to the assumption on agent~$i_t$, so agent~$i_t$ also remains EF2.
By similar reasoning, for a simple directed cycle $i_0\rightarrow i_1\rightarrow\dots\rightarrow i_t\rightarrow i_0$, if we start the transfer process by transferring a good from an agent with the largest bundle, every intermediate allocation satisfies EF2.

For a clean allocation~$\cA$, we define the potential\footnote{This potential function was also used by \citet{BenabbouChIg21}.} $\Phi(\cA) = \sum_{i\in N} d_i^2$.
A directed path $i_0\rightarrow\dots\rightarrow i_t$ is called \emph{improving} if $d_{i_t} \ge d_{i_0} + 2$.
Observe that implementing transfers along an improving path changes the potential by
\[
(d_{i_t} - 1)^2 + (d_{i_0} + 1)^2 - d_{i_t}^2 - d_{i_0}^2 
= -2(d_{i_t} - d_{i_0} - 1) < 0,
\]
that is, it decreases the potential.
Hence, a clean allocation that minimizes the potential among all clean allocations cannot contain an improving path.
We claim that the converse also holds: if a clean allocation does not contain any improving path, then it minimizes the potential (among all clean allocations).
Let $\cA$ be a clean allocation without an improving path, and let $\cB$ be any clean allocation.
For each $i\in N$, let $e_i = |B_i|$ and $\delta_i = e_i - d_i$.
Construct a \emph{difference (multi)graph} where the set of vertices corresponds to the set of agents~$N$, and for each good whose owners differ in $\cA$ and $\cB$, we add an edge from its owner in~$\cA$ to its owner in~$\cB$.
Decompose this graph into simple directed paths from surplus agents (i.e., those for whom $d_i > e_i$) to deficit agents (i.e., those for whom $d_i < e_i$), along with simple directed cycles.
Let $\cP$ be the collection of paths, counted with multiplicity.
Each surplus agent~$i$ starts $d_i-e_i$ paths, while each deficit agent~$i$ ends $e_i-d_i$ paths.

For a path $P\in \cP$, let $a(P)$ and $b(P)$ denote its starting and ending agents, respectively.
Then, $\sum_{i\in N}|\delta_i| = 2|\cP|$ and 
\[
\sum_{i\in N} d_i\delta_i 
= \sum_{i\in N} d_i(\text{indeg}(i) - \text{outdeg}(i))
= \sum_{P\in \cP} (d_{b(P)} - d_{a(P)}),
\]
where the indegrees and outdegrees are with respect to the difference graph.
Hence,
\begin{align}
\Phi(\cB) - \Phi(\cA)
&= \sum_{i\in N} e_i^2 - \sum_{i\in N} d_i^2 \nonumber \\
&= \sum_{i\in N} (d_i + \delta_i)^2 - \sum_{i\in N} d_i^2 \nonumber \\
&= \sum_{i\in N} 2d_i\delta_i + \sum_{i\in N} \delta_i^2 \nonumber \\
&= 2\sum_{P\in \cP} (d_{b(P)} - d_{a(P)}) + \sum_{i\in N} \delta_i^2 \nonumber \\
&= 2\sum_{P\in \cP} (d_{b(P)} - d_{a(P)} + 1) + \sum_{i\in N} (\delta_i^2 - |\delta_i|). \label{eq:potential-difference}
\end{align}
Every term in the last expression is nonnegative.
Indeed, $\delta_i^2 - |\delta_i| \ge 0$ holds because $\delta_i$ is an integer.
If $d_{b(P)} - d_{a(P)} + 1 < 0$ for some path~$P\in\cP$, then $d_{a(P)} - d_{b(P)} \ge 2$; since $\cA$ and $\cB$ are clean allocations, this would give rise to an improving path in the transfer graph (of~$\cA$) from $b(P)$ to~$a(P)$, a contradiction.
It follows that $\Phi(\cA) \le \Phi(\cB)$, so $\cA$ indeed minimizes the potential among all clean allocations.
Moreover, $\cA$ must also be EF1.
To see this, note that if $u_i(A_j) > 0$ for some distinct agents $i$ and $j$, then there is an edge $i\rightarrow j$ in the transfer graph.
The absence of a one-edge improving path implies that $u_i(A_j) \le |A_j| = d_j \le d_i + 1 = u_i(A_i) + 1$.

Next, we show that any clean EF2 allocation~$\cA$ can be connected to a clean allocation minimizing~$\Phi$ via transfers through clean EF2 allocations.
Assume that the current clean EF2 allocation~$\cA$ does not minimize $\Phi$.
Suppose first that some agent holds a good valued by an agent whose utility is lower by at least~$2$.
Among all agents who can give away a good in such a transfer, choose an agent~$i^*$ with the largest~$d_{i^*}$, and choose a good~$g$ that allows such a transfer.
Among all agents who can receive~$g$, choose an agent~$j^*$ with the minimum utility; in particular, $d_{i^*} \ge d_{j^*} + 2$.
Let $\cB$ be the allocation obtained by transferring $g$ from~$i^*$ to~$j^*$.
For any agent $i\ne i^*$, her own utility does not decrease, and the only bundle for which her utility can increase is that of agent~$j^*$, when agent~$i$ values~$g$.
If agent~$i$ values~$g$, then the choice of $j^*$ implies that $d_i \ge d_{j^*}$, and so $u_i(B_{j^*}) \le |B_{j^*}| = d_{j^*} + 1 \le d_i+1$, so $i$ remains EF2.
For agent~$i^*$, before the transfer, we have $u_{i^*}(A_j) \le d_{i^*} + 1$ for every $j\in N$---otherwise, for an agent~$j$ where this is violated, we would have $d_j = |A_j| \ge u_{i^*}(A_j) \ge d_{i^*}+2$, so agent~$j$ can transfer a good to agent~$i^*$, contradicting the maximality of $d_{i^*}$.
After the transfer from~$\cA$ to~$\cB$, the utility of agent~$i^*$ is $d_{i^*} - 1$.
For each unchanged bundle, the utility of agent~$i^*$ is at most $d_{i^*} + 1$, and the enlarged bundle of agent~$j^*$ has size $d_{j^*}+1 \le d_{i^*} - 1$.
Hence, $\cB$ is a clean EF2 allocation, and the transfer strictly decreases $\Phi$ since $d_{i^*} \ge d_{j^*} + 2$.

Suppose now that in~$\cA$, no agent holds a good valued by an agent whose utility is lower by at least~$2$.
Hence, for each edge $i\rightarrow j$ in the transfer graph, we have $d_j \le d_i + 1$, and so $\cA$ is EF1.
Since $\cA$ is not a potential minimizer, the transfer graph contains an improving path $i_0\rightarrow i_1\rightarrow\dots\rightarrow i_t$, where $d_{i_t} \ge d_{i_0} + 2$.
Let $i_r$ be a maximum-utility agent on this path, and retain the prefix $i_0\rightarrow i_1\rightarrow\dots\rightarrow i_r$; this prefix is still an improving path since $d_{i_r} \ge d_{i_t} \ge d_{i_0}+2$.
Since $\cA$ is EF1, as argued earlier, we can transfer goods backward along this path so that every intermediate allocation is EF2.
Again, this transfer process strictly decreases~$\Phi$.
By repeating such transfer processes, we eventually reach a clean EF2 allocation minimizing~$\Phi$.

It remains to connect any two clean allocations minimizing~$\Phi$ via EF2 allocations.
As shown earlier, clean allocations minimizing~$\Phi$ must be EF1.\footnote{This is also Corollary~3.8 in the work of \citet{BenabbouChIg21}.}
Let $\cA$ be the initial clean allocation and $\cB$ the target clean allocation.
Define the difference graph between $\cA$ and $\cB$ as before, and decompose it into simple directed paths from surplus agents to deficit agents, along with simple directed cycles.
Since both $\cA$ and $\cB$ minimize~$\Phi$, and all terms on the right-hand side of~\eqref{eq:potential-difference} are nonnegative, all of these terms must be $0$.
Hence, for any path $i_t\rightarrow i_{t-1}\rightarrow \dots \rightarrow i_0$ in the chosen decomposition, it holds that $d_{i_t} = d_{i_0}+1$.
If $d_{i_r} > d_{i_t}$ for some $r\in\{1,\dots,t-1\}$, we would have $d_{i_r} \ge d_{i_0} + 2$, and $i_0\rightarrow \dots \rightarrow i_r$ would be an improving path in the transfer graph of~$\cA$, a contradiction.
Hence $d_{i_t} \ge d_{i_r}$ for all $r\in\{0,1,\dots,t\}$.
Therefore, as argued earlier, we can transfer goods backward along the path $i_0\rightarrow\dots\rightarrow i_t$ in the transfer graph so that every intermediate allocation is EF2.
After the transfers, agent~$i_t$'s utility decreases by~$1$, agent~$i_0$'s utility increases by~$1$, and every other agent's utility remains the same.
Since $d_{i_t} = d_{i_0}+1$, the resulting allocation also minimizes the potential~$\Phi$ among all clean allocations and is therefore EF1.
Moreover, each transferred good is with its target owner.
In the case of a simple cycle, we can execute the transfers starting with a maximum-utility agent, and EF2 is maintained throughout.
Since no agent's utility changes, the resulting allocation again minimizes~$\Phi$ and is EF1.
It follows that we can reach $\cB$ from $\cA$ via an EF2 reconfiguration path, as desired.
\end{proof}

Finally, in a similar vein as \Cref{thm:recursively-balanced-divisible}, we show that if the initial and target allocations are recursively balanced, then there exists a reconfiguration path with the same property.
In fact, with transfers allowed, we do not need the condition that the number of goods is divisible by the number of agents.

\begin{theorem}
\label{thm:transfer-recursively-balanced}
In the exchange-and-transfer setting, for any instance and any initial and target recursively balanced allocations, there exists a reconfiguration path such that every intermediate allocation is recursively balanced (and therefore EF1). 
\end{theorem}

\begin{proof}
Consider any instance and any initial and target recursively balanced allocations, $\cA$ and $\cB$.
If $m$ is divisible by~$n$, the result already follows from \Cref{thm:recursively-balanced-divisible} with exchanges alone.
Assume from now on that $m$ is not divisible by~$n$, and write $m = qn + r$ where $r\in\{1,\dots,n-1\}$.
Construct an auxiliary instance with the $m$ real goods together with $n-r$ dummy goods, so that the total number of goods is $m+(n-r) = (q+1)n$.
For a real good~$g$ and any agent~$i$, the auxiliary utility is $\widehat{u}_i(g) = u_i(g) + 1$.
All agents have utility~$0$ for all dummy goods.

Consider the recursively balanced picking sequences (and tie-breaking choices, if any) that yield $\cA$ and $\cB$ in the original instance.
If we run the same picking sequences on the auxiliary instance, the first $m$~picks remain the same.
Let each of the $n-r$ agents who have not picked in the last round take an arbitrary dummy good (this set of agents may differ between $\cA$ and $\cB$), and call the respective resulting allocations $\widehat{\cA}$ and $\widehat{\cB}$.
\Cref{thm:recursively-balanced-divisible} yields an exchange-only reconfiguration path between $\widehat{\cA}$ and $\widehat{\cB}$ consisting only of recursively balanced allocations in the auxiliary instance.
Since agents cannot pick dummy goods ahead of real goods, for each of these allocations, the dummy goods are picked in the last $n-r$ picks and each agent can have at most one dummy good.

Now, remove the dummy goods from each allocation on the reconfiguration path.
Since the order of utilities for the real goods (including ties) is the same in the original and auxiliary instances, each resulting allocation can be produced by a recursively balanced picking sequence in the original instance.
Moreover, an exchange between two dummy goods can be ignored, while an exchange between a dummy good and a real good can be turned into a transfer.
Hence, this gives rise to an exchange-and-transfer reconfiguration path for the original instance, as desired.
\end{proof}

\section*{Acknowledgments}

This work is partially supported by the Deutsche Forschungsgemeinschaft under grant BR 2312/14-1, by the Ministry of Education, Singapore, under the Academic Research Fund Tier 1 (FY 2026) grant number 251RES2604, and by an NUS Start-up Grant. 
ChatGPT 5.6 Pro was used to help devise proofs with the authors' guidance.
The authors wrote the entire paper and take full responsibility for its content.

\bibliographystyle{plainnat} 
\bibliography{main}

\end{document}